\documentclass[onecolumn,12pt, journal,draftclsnofoot]{IEEEtran}
\usepackage[T1]{fontenc}% optional T1 font encoding

\usepackage{bm}
\usepackage{amsthm}
\usepackage{amsmath}
\newtheorem{lemma}{Lemma}
\newtheorem{theorem}{Theorem}
\newtheorem{corollary}{Corollary}

\usepackage{cite}
\usepackage[pdftex]{graphicx}

\ifCLASSINFOpdf
\else
\fi
\usepackage{diagbox}

\usepackage{amsmath}

\usepackage[cmintegrals]{newtxmath}
\usepackage{algorithm}

\usepackage{algcompatible}
\usepackage{hyperref}
\usepackage{subfigure}

\begin{document}
\title{Grant-free Massive Random Access with Retransmission: Receiver Optimization and Performance Analysis}

\author{Xinyu~Bian,~\IEEEmembership{Student Member,~IEEE,}
        Yuyi~Mao,~\IEEEmembership{Member,~IEEE,}
        and~Jun~Zhang,~\IEEEmembership{Fellow,~IEEE}% <-this % stops a space
\thanks{X. Bian and J. Zhang are with the Department of Electronic and Computer Engineering, the Hong Kong University of Science and Technology, Hong Kong (E-mail: xinyu.bian@connect.ust.hk, eejzhang@ust.hk). Y. Mao is with the Department of Electronic and Information Engineering, the Hong Kong Polytechnic University, Hong Kong (E-mail: yuyi-eie.mao@polyu.edu.hk). This work was supported by the General Research Fund (Project No. 15207220) from the Hong Kong Research Grants Council. \emph{(Corresponding author: Yuyi Mao)}

Part of this work was presented at the 2022 IEEE Global Communications Conference (GLOBECOM) \cite{xbian2022}.}% <-this % stops a space
\thanks{}% <-this % stops a space
\thanks{}}


\markboth{}%
{Shell \MakeLowercase{\textit{et al.}}: Bare Demo of IEEEtran.cls for IEEE Communications Society Journals}



\maketitle

\begin{abstract}
There is an increasing demand of massive machine-type communication (mMTC) to provide scalable access for a large number of devices, which has prompted extensive investigation on grant-free massive random access (RA) in 5G and beyond wireless networks. Although many efficient signal processing algorithms have been developed, the limited radio resource for pilot transmission in grant-free massive RA systems makes accurate user activity detection and channel estimation challenging, which thereby compromises the communication reliability. In this paper, we adopt retransmission as a means to improve the quality of service (QoS) for grant-free massive RA. Specifically, by jointly leveraging the user activity correlation between adjacent transmission blocks and the historical channel estimation results, we first develop an activity-correlation-aware receiver for grant-free massive RA systems with retransmission based on the correlated approximate message passing (AMP) algorithm. Then, we analyze the performance of the proposed receiver, including the user activity detection, channel estimation, and data error, by resorting to the state evolution of the correlated AMP algorithm and the random matrix theory (RMT). Our analysis admits a tight closed-form approximation for frame error rate (FER) evaluation. Simulation results corroborate our theoretical analysis and demonstrate the effectiveness of the proposed receiver for grant-free massive RA with retransmission, compared with a conventional design that disregards the critical user activity correlation.
\end{abstract}

\begin{IEEEkeywords}
Grant-free massive random access (RA), retransmission, user activity detection, channel estimation, approximate message passing (AMP), finite-blocklength coding, frame error rate (FER).
\end{IEEEkeywords}

% \fi
\IEEEpeerreviewmaketitle

\section{Introduction}
The proliferation of the Internet of Things (IoT) has enabled tremendous thrilling applications, such as wearable electronics, smart cities, and intelligent manufacturing, which have fundamentally transformed our everyday lives \cite{aal2015}. Meanwhile, the exponential increase of IoT devices \cite{sta2016} has imposed an unprecedented challenge to future wireless networks in supporting massive concurrent connections. Hence, massive machine-type communications (mMTC) has been identified as one of the core services in the fifth generation (5G) and beyond wireless networks \cite{itu}. A unique feature in mMTC is that only a small proportion of the massive users simultaneously transmit short data packets at a given instant \cite{cbo2016}. This feature warrants a major paradigm shift of RA schemes, from the conventional grant-based RA, where each user requires a permission from the base station (BS) before transmission, to the grant-free RA, where each active user directly sends the payload data without waiting for the BS's approval \cite{xchen2021,yshi2020, xbian2021conference}. 

Although grant-free RA is prominent in reducing the access latency and signaling overhead for mMTC, the limited radio resource for pilot transmission makes it difficult to perform accurate user activity detection and channel estimation at the BS, which is critical to the downstream data detection and decoding. Many efficient signal processing algorithms have been developed to improve the accuracy of user activity detection and channel estimation \cite{donoho2006}, but uplink data transmission in grant-free massive RA systems may still be prone to failure with only a single transmission attempt. To improve the communication reliability, in this paper, we advocate to integrate the existing grant-free massive RA systems with a retransmission mechanism. Such a mission is highly nontrivial despite retransmission protocols such as the hybrid automatic repeat request (HARQ) have a long history for grant-based RA \cite{aah2021}. Specifically, since the unknown user activity remains static across multiple transmission blocks in grant-free massive RA, a new receiver capable of exploiting such user activity correlation is imperative in order to maximize the benefits of retransmission for reliable mMTC.

\subsection{Related Works and Motivations}
There has been a recent growing interest in grant-free massive RA. To achieve activity detection for grant-free massive multi-input multi-output (MIMO) systems, a low-complexity algorithm was developed through the sample covariance matrix of the received pilot signal in \cite{shag2018}. Such an algorithm, however, is not able to estimate the unknown channel coefficients, which are critical to data detection and decoding. Therefore, joint activity detection and channel estimation (JADCE) has received significant attention in subsequent investigations \cite{zchen2018,mke2020,qzou2020,xbian2021}. In particular, by formulating JADCE as a compressive sensing problem, an approximate message passing (AMP)-based algorithm was proposed in \cite{zchen2018}, which exploits the wireless channel statistics and provides a theoretical performance characterization. In \cite{mke2020}, both the spatial and angular domain channel sparsity are leveraged to improve the accuracy via alternatively detecting the user activity and estimating the channel coefficients. Besides, motivated by the common sparsity in pilot and data signal, data-assisted approaches have been most recently developed for JADCE. For example, a joint activity detection, channel estimation, and multi-user detection algorithm was proposed in \cite{qzou2020} under the framework of bilinear generalized AMP (BiG-AMP), which utilizes the estimated payload data to aid channel estimation. In addition, the data decoding procedures were jointly optimized with JADCE in \cite{xbian2021}, where the extrinsic information derived from a channel decoder is used as the auxiliary prior for JADCE. Nonetheless, these works only focus on grant-free massive RA systems with one-shot data transmission, which neglect the potentials of retransmission mechanisms to secure further performance improvements.

As an indispensable mechanism in modern digital communication systems, retransmission is advantageous for grant-free massive RA to achieve reliable communications as it provides more diversity gain. In \cite{fjab2021}, a variety of HARQ schemes were adapted for grant-free non-orthogonal multiple access (NOMA) systems, where the pilot symbols are not merely used for JADCE, but also for carrying the HARQ meta-data. Authors of \cite{hgsr2022} proposed a novel grant-free RA mechanism to resolve collisions, where users are allowed to transmit in one or multiple opportunities chosen at random. The probability of successful transmission was analyzed in closed form, which unveils the fundamental tradeoffs between the number of preambles and retransmissions in terms of latency and energy consumption. Likewise, \cite{dma2022} considered a general retransmission-based RA framework that uniﬁes NOMA with HARQ-based protocols, and characterized the spectral efﬁciency/user density versus per-bit signal-to-noise ratio (SNR) tradeoffs. Moreover, in order to improve the throughput of grant-free RA with massive MIMO, immediate retransmissions are initiated for active users with a low signal-to-interference-plus-noise ratio (SINR) in \cite{jchoi2020}. However, all the aforementioned works exploited the retransmission diversity for grant-free massive RA simply through conventional HARQ operations, e.g., chase combining and incremental redundancy, leaving the signal processing algorithms at the BS largely intact. A key observation in grant-free massive RA is that the user activity remains static across adjacent transmission blocks, which is valuable for improving the performance of JADCE and thus the quality of service (QoS). Developing a novel receiver dedicated for grant-free massive RA systems with retransmission by leveraging such user activity correlation and analyzing its performance form the main motivations of this paper.

%In \cite{fghan2021}, the packet error rate (PER) and throughput of each user has been investigated for non-orthogonal multiple access (NOMA) with hybrid automatic repeat request (HARQ) systems by using the Markov models, and based on it, two different optimization problems were solved numerically for the power constrained and reliability constrained scenarios.

\subsection{Contributions}
In this paper, we endeavor to develop a new receiver for grant-free massive RA systems with retransmission and quantify its performance through theoretical analysis. Our main contributions are summarized below.

\begin{itemize}
    \item To improve the reliability of mMTC, we investigate a typical grant-free massive RA system with retransmission. Noticing that the user activity remains static across multiple transmission blocks and the previous channel estimation results are informative for JADCE in the next transmission block, we develop a novel activity-correlation-aware receiver based on the correlated AMP algorithm \cite{annama2019} to leverage both types of information, where the historical information (HI) is judiciously designed based on the state evolution.
\end{itemize}

\begin{itemize}
    \item We first characterize the user activity detection error of the proposed receiver, including the false alarm and missed detection probabilities, based on, again, the state evolution of the correlated AMP algorithm. Then, conditioned on the outcome of user activity detection, we derive the block error rate (BLER) expression in the finite block-length regime by analyzing the distribution of the receive signal-to-noise ratio (SNR) via the random matrix theory (RMT), which also admits a tight closed-form approximation. Finally, the frame error rate (FER) expression is obtained in a recursive form based on the BLER analysis.
\end{itemize}

\begin{itemize}
    \item The extensive simulation results first corroborate our theoretical performance analysis. They also demonstrate the effectiveness of the proposed activity-correlation-aware receiver over the conventional AMP-based design for grant-free massive RA systems with retransmission. The performance improvement achieved by the proposed receiver becomes more prominent for data traffic that can tolerate a longer delay.
\end{itemize}

\subsection{Organization}
The rest of this paper is organized as follows. We introduce the grant-free massive RA system with retransmission in Section \ref{sectionii}. In Section \ref{sectioniii}, we review the conventional AMP-based receiver for grant-free massive RA systems. A novel activity-correlation-aware receiver is developed in Section \ref{sectioniv} and its performance is analyzed in Section \ref{sectionv}. Simulation results are presented in Section \ref{sectionvi} and Section \ref{sectionvii} concludes this paper.

\subsection{Notations} 
We use lower-case letters, bold-face lower-case letters, bold-face upper-case letters, and math calligraphy letters to denote scalars, vectors, matrices, and sets, respectively. The matrix inverse, transpose, and conjugate transpose operators are denoted as $(\cdot)^{-1}$, $(\cdot)^{\mathrm{T}}$, and $(\cdot)^{\mathrm{H}}$, respectively. In addition, $\delta_{0}\left(\cdot\right)$ denotes the Dirac delta function, and $\mathcal{CN}(\boldsymbol{\mu},\mathbf{v})$ denotes the probability density function (PDF) of a complex Gaussian random variable with mean $\boldsymbol{\mu}$ and variance $\mathbf{v}$.

\section{System Model \label{sectionii}}
\subsection{Grant-free Massive RA Systems with Retransmission}
We investigate the uplink communication procedures in grant-free massive RA systems with retransmission. As shown in Fig. \ref{model}, we assume $N$ single-antenna users (denoted as set $\mathcal{N} \triangleq \{1,\cdots,N\}$) with delay-sensitive traffic are served by an $M$-antenna BS. The delay requirements of users are assumed equal to the length of a frame that consists of $J$ ($J \geq 2$) transmission blocks. At the beginning of each frame, $K$ out of the $N$ users become active for short-packet transmission, and the set of active users at the beginning of a typical frame is denoted as $\mathcal{K}$. Hence, a data packet can undergo at most $J$ transmission attempts, including one initial transmission and $J-1$ retransmissions. To avoid the system from being overloaded \cite{b14}, the number of BS antennas is assumed to be no less than the number of active users, i.e., $M \geq K$. 

For ease of exposition, the stop-and-wait automatic repeat request (ARQ) protocol is adopted to control the transmission process, where acknowledgements (ACKs) are fed back to users by the BS after their data packets are correctly received. Therefore, if an active user receives an ACK message before timeout (i.e., the end of the frame), it succeeds in transmission and becomes inactive in the remaining transmission blocks. Otherwise, it retransmits the packet until all the $J-1$ retransmission attempts are exhausted. We denote $K^{(j)}$ ($K^{(j)}\leq K$) as the number of active users in the $j$-th transmission block, which is known at the BS under stop-and-wait ARQ. The set of active users in the $j$-th transmission block is represented by $\mathcal{K}^{(j)} \triangleq \left\{n \in \mathcal{K} | u_{n}^{(j)}=1 \right\}$, where $u_{n}^{(j)}\in\{0,1\}$ is the user activity indicator with $u_{n}^{(j)} = 1$ meaning user $n$ is active and vice versa. Notably, $\mathcal{K}^{1} = \mathcal{K}$, $K^{\left(J\right)} \geq K^{\left(J-1\right)} \geq \cdots \geq K^{\left(1\right)}$, and the activity status of a user in $\mathcal{K}$ does not change within in a frame unless an ACK message is received. Besides, if a data packet cannot be successfully decoded within a frame, it is discarded in the new frame and a transmission error occurs.
\begin{figure}[htpb]
\centering
\includegraphics[width=4.4in]{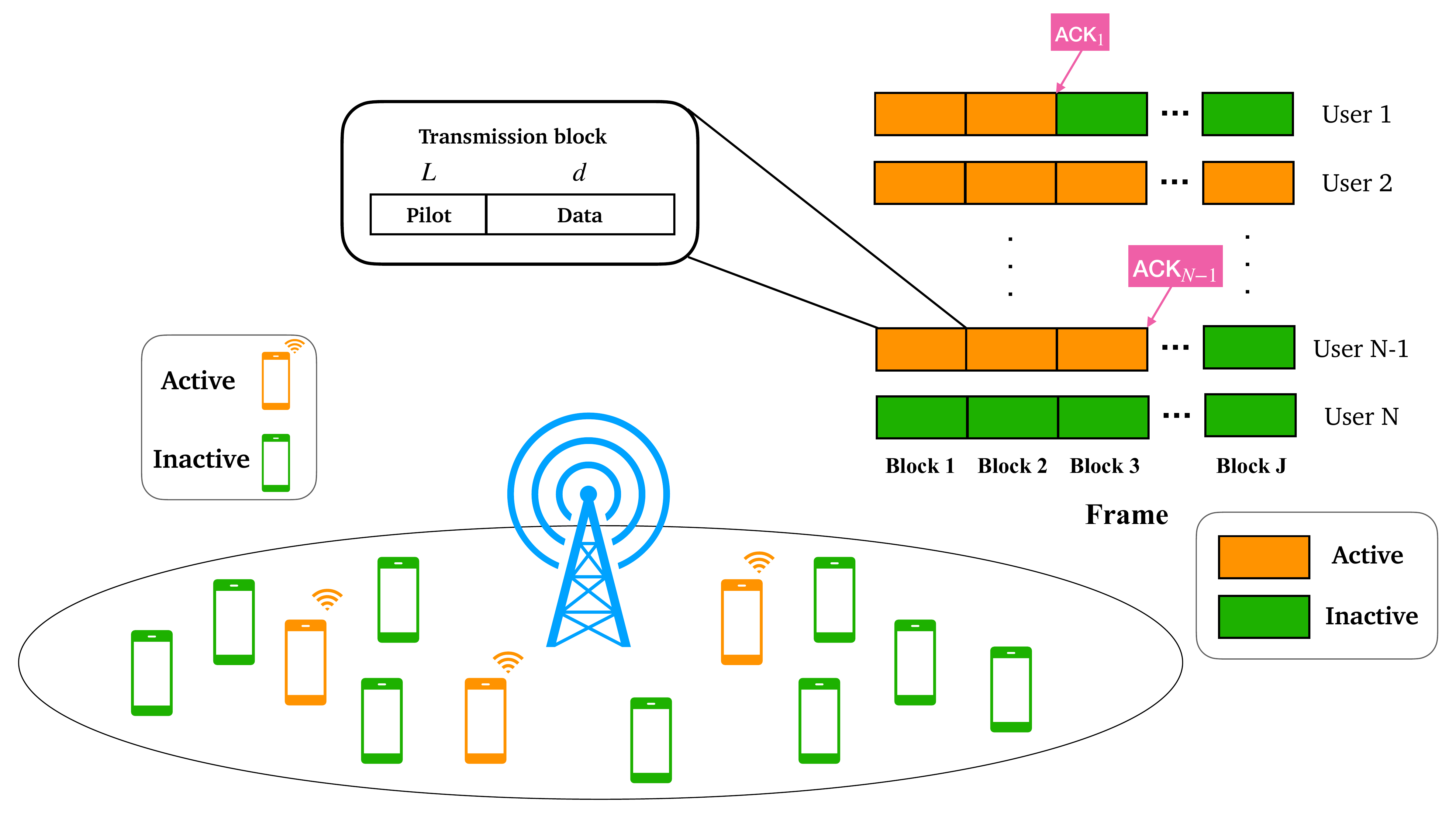}
\caption{Model and frame structure of a grant-free massive RA system with retransmission. In each frame, an active user has $J$ attempts to transmit its data. Once an ACK message is received from the BS, the active user becomes inactive. Otherwise, it remains active to retransmit data until the last transmission block.}
\label{model}
\end{figure}

\subsection{Signal Model}
The block fading channel model is adopted, where the channel condition remains unchanged within a transmission block spanning $T$ symbol intervals, but varies independently across different transmission blocks. In the $j$-th transmission block, the uplink channel vector from user $n$ to the BS is denoted as $\mathbf{f}_{n}^{(j)}=\sqrt{\beta_{n}} \boldsymbol{\alpha}_{n}^{(j)}$, where $\boldsymbol{\alpha}_{n}^{(j)}$ and $\beta_{n}$ stand for the small-scale and large-scale fading coefficients, respectively. We focus on the Rayleigh fading channels, i.e., $\boldsymbol{\alpha}_{n}^{(j)} \sim \mathcal{CN}(\boldsymbol{0},\mathbf{I}_{M})$, and assume $\{\beta_n\}$'s are known at the BS.

The classic two-phase grant-free RA scheme \cite{zchen2018} is implemented in each transmission block as shown in Fig. \ref{model}. In the first phase, $L$ symbols are reserved for pilot transmission, and $d\triangleq T-L$ symbols are used for data delivery in the second phase. Since the number of users can be much larger than the pilot length, assigning orthogonal pilots to all the users is impractical. Therefore, a unique pilot sequence $\sqrt{L}\mathbf{p}_{n}$ is allocated to each user, where $\mathbf{p}_{n}\triangleq \left[p_{n, 1}, \cdots, p_{n, L}\right]^{\mathrm{T}}$ and $p_{n, l} \sim \mathcal{C} \mathcal{N}\left(0, 1/L\right)$. Despite being non-orthogonal, the set of pilot sequences achieves asymptotic orthogonality when $L\rightarrow \infty$ \cite{zchen2018}. Besides, the statistical channel inversion power control \cite{b15} is implemented, i.e., the transmit power $\rho_{n}$ for each user is chosen such that $\rho_{n}\beta_{n}=\beta$, where $\beta$ is the target received signal strength at the BS. Then, the received pilot signal $\mathbf{Y}_{p}^{(j)} \in \mathbb{C}^{L\times M}$ at the BS in the $j$-th transmission block can be expressed as follows: 
\begin{align}
\mathbf{Y}_{p}^{(j)}=\sqrt{L}\mathbf{P}\mathbf{H}^{(j)}+\mathbf{N}_{p}^{(j)}=\sum_{n \in \mathcal{N}}\sqrt{L}\mathbf{p}_{n}(\mathbf{h}^{(j)}_{n})^{\mathrm{T}}+\mathbf{N}_{p}^{(j)},
\end{align}
\noindent where $\mathbf{P}\triangleq \left[\mathbf{p}_{1},\cdots,\mathbf{p}_{N}\right]$, $\mathbf{H}^{(j)} \triangleq \left[\mathbf{h}^{(j)}_{1},...,\mathbf{h}^{(j)}_{N}\right]^{\mathrm{T}}$ represents the effective channel matrix in the $j$-th transmission block with $\mathbf{h}^{(j)}_{n}\triangleq u_{n}^{(j)}\sqrt{\beta}\boldsymbol{\alpha}_{n}^{(j)}$, and $\mathbf{N}_{p}^{(j)}\triangleq \left[\mathbf{n}_{p,1}^{(j)},...,\mathbf{n}_{p,L}^{(j)}\right]^{\mathrm{T}}$ denotes the Gaussian noise with zero mean and variance $\sigma^{2}$ for each element.

A short data packet with $c$ payload bits needs to be transmitted for each active user, which is encoded as a codeword of $d$ symbols, denoted as $\mathbf{s}_{n}$. We further assume each symbol in a codeword has unit power. Therefore, the received data signal at the BS in the $j$-th transmission block $\mathbf{Y}^{(j)}\in \mathbb{C}^{M \times d}$ can be expressed as follows: 
\begin{align}
\mathbf{Y}^{(j)}= \sum_{n \in \mathcal{K}^{(j)}} \mathbf{h}_{n}^{(j)} \mathbf{s}_{n}^{\mathrm{T}}+\mathbf{N}^{(j)},
\end{align}

\noindent where the noise term $\mathbf{N}^{(j)} \triangleq \left[\mathbf{n}_{1}^{(j)},...,\mathbf{n}_{d}^{(j)}\right]$ follows the same distribution as $\mathbf{N}_{p}^{(j)}$.

\section{AMP-based Receiver for Grant-Free Massive RA \label{sectioniii}}
The AMP-based receiver \cite{lliu2018} is a popular choice for grant-free massive RA, which is responsible for activity detection, channel estimation, data detection and decoding for the active users. In particular, based on the received pilot signal in each transmission block, it first performs the AMP algorithm \cite{zchen2018} for joint activity detection and channel estimation. Then, the conventional data detection and decoding procedures are applied for users that are detected as active. In this section, we review the basic principles of the AMP algorithm for JADCE, as well as data detection and decoding.

\subsection{AMP for JADCE}
In the first phase of each transmission block, the AMP algorithm is adopted to process the received pilot signal $\mathbf{Y}_{p}^{(j)}$ at the BS to estimate the set of active users and their channel coefficients. This algorithm not only exhibits superior performance, and its state evolution \cite{lliu2018} also provides a powerful framework for the performance analysis of grant-free massive RA systems. In particular, in the $j$-th transmission block, the AMP algorithm starts from $\mathbf{H}_{0}^{(j)}=\mathbf{0}$ and $\mathbf{R}_{0}^{(j)}=\mathbf{Y}_{p}$, and iterates as follows:
\begin{align}
    \mathbf{h}_{n, t+1}^{(j)}=\eta_{t}^{(j)}\left(\mathbf{h}_{n,t}^{(j)}+\left(\mathbf{R}_{t}^{(j)}\right)^{\mathrm{H}} \mathbf{p}_{n}\right),
\label{esth}
\end{align}
\begin{align}
\mathbf{R}_{t+1}^{(j)} =\mathbf{Y}_{p}^{(j)}-\mathbf{P} \mathbf{H}_{t+1}^{(j)}+\frac{N}{L} \mathbf{R}_{t}^{(j)}\sum_{n=1}^{N}\frac{\eta_{t}^{\prime (j)}\left(\mathbf{h}_{n,t}^{(j)}+\left(\mathbf{R}_{t}^{(j)}\right)^{\mathrm{H}} \mathbf{p}_{n}\right)}{N},
\end{align}

\noindent where $t=0,1,...$ denotes the iteration index, $\mathbf{H}_{t}^{(j)}\triangleq \left[\mathbf{h}_{1,t}^{(j)},...,\mathbf{h}_{N,t}^{(j)}\right]^{\mathrm{T}}$ is the estimate of $\mathbf{H}^{(j)}$ in the $t$-th iteration, and $\mathbf{R}_{t}^{(j)}$ represents the residual of the received pilot signal in iteration $t$. Besides, $\eta_{t}^{(j)}(\cdot)$ as expressed below is the minimum mean square error (MMSE) denoiser that provides an estimate of the effective channel coefficients \cite{zchen2018}:
\begin{align}
    \eta_{t}^{(j)}(\mathbf{a}_{n,t}^{(j)})=\frac{\frac{\beta\mathbf{a}_{n,t}^{(j)}}{\beta+(\tau_t^{(j)})^2}}{1+\frac{1-\lambda^{(j)}}{\lambda^{(j)}} \left(\frac{\beta+(\tau_t^{(j)})^2}{(\tau_t^{(j)})^2}\right)^M e^{\left(\frac{1}{\beta+(\tau_t^{(j)})^2}-\frac{1}{(\tau_t^{(j)})^2}\right)||\mathbf{a}_{n,t}^{(j)}||^2}},
\label{denoiser}
\end{align}

\noindent where $\mathbf{a}_{n,t}^{(j)}\triangleq \mathbf{h}_{n}^{(j)}+(\tau_{t}^{(j)})^2 \mathbf{v}_{n}^{(j)}$ with $\mathbf{v}_{n}^{(j)}\sim \mathcal{CN}(\mathbf{0},\mathbf{I})$ denoting the equivalent noise independent of $\mathbf{h}_{n}^{(j)}$, and $\eta_{t}^{\prime (j)}(\cdot)$ is the ﬁrst-order derivative of $\eta_{t}^{(j)}(\cdot)$. In addition, $\lambda^{(j)} = K^{(j)} \slash N$ and $\tau_t^{(j)}$ denotes the AMP state in the $t$-th iteration, which follows the evolution as stated in Lemma \ref{lemma1}.

%which is the equivalent model that characterize the distribution of applying the denoiser to $\mathbf{h}_{n,t}^{j}+\left(\mathbf{R}_{t}^{j}\right)^{\mathrm{H}} \mathbf{p}_{n}$
\begin{lemma} \label{lemma1}
[13, Proposition 9]
Define two random variables $\mathbf{X}^{(j)}$ and $\mathbf{V}^{(j)}$ following distributions $(1-\lambda^{(j)}) \delta_{0}+\lambda^{(j)} \mathcal{CN}(0,\beta\mathbf{I})$ and $\mathcal{CN}(0,\mathbf{I})$, respectively. In the asymptotic regime where $N,L,K^{(j)} \rightarrow \infty$, the state evolution of the AMP algorithm with the MMSE denoiser in (\ref{denoiser}), which tracks the mean square error (MSE) of the estimated effective channel coefficients, exhibits the following recursion:
\begin{align}
\mathbf{\Sigma}_{t+1}^{(j)}=\frac{\sigma^{2}}{L} \mathbf{I}+\frac{N}{L} \mathbb{E}_{\mathbf{X}^{(j)},\mathbf{V}^{(j)}}\left[\Big |\Big|\eta_{t}^{(j)}(\mathbf{X}^{(j)}+\mathbf{(\Sigma}_{t}^{(j)})^{\frac{1}{2}} \mathbf{V}^{(j)})-\mathbf{X}^{(j)}\Big|\Big|_{2}^{2}\right].
\end{align}

\noindent In the asymptotic regime, the state matrix stays as a scaled identity matrix when the MMSE denoiser is applied, i.e., $\mathbf{\Sigma}_{t}^{(j)}=(\tau_{t}^{(j)})^{2}\mathbf{I}$.
\end{lemma}
%and $\mathbf{\Sigma}_{0}^{j}=\frac{\sigma^{2}}{L} \mathbf{I}+\frac{N}{L} \mathbb{E}_{\mathbf{X}^{j}}\left[\mathbf{X}^{j}(\mathbf{X}^{j})^{\mathrm{H}}\right]$.

Hence, upon the AMP algorithm converges, an estimate of $\{\mathbf{h}_{n}^{(j)}\}$, denoted as $\{\hat{\mathbf{h}}_{n}^{(j)}\}$, can be obtained as $\hat{\mathbf{h}}_{n}^{(j)}=\mathbf{h}_{n,\infty}^{(j)}$. By further applying a thresholding operation, the set of active users is determined as $\hat{\mathcal{K}}^{(j)}=\{i \in \mathcal{N}| \hat{u}_{i}^{(j)}=1\}$, with $\hat{u}_{i}^{(j)}$ given as follows:
\begin{align}
\hat{u}_{i}^{(j)}=\left\{\begin{array}{l}
1, \text { if }\lVert\mathbf{h}_{n,\infty}^{(j)}+\left(\mathbf{R}_{\infty}^{(j)}\right)^{\mathrm{H}}\mathbf{p}_{n}\rVert^{2}\geq l^{(j)} \\
0, \text { if }\lVert\mathbf{h}_{n,\infty}^{(j)}+\left(\mathbf{R}_{\infty}^{(j)}\right)^{\mathrm{H}}\mathbf{p}_{n}\rVert^{2}< l^{(j)}
\end{array}\right.,
\label{activeamp}
\end{align}

\noindent where $l^{(j)}\triangleq M \ln \left(1+\frac{\beta}{(\tau_{\infty}^{(j)})^{2}}\right) /\left(\frac{1}{(\tau_{\infty}^{(j)})^{2}}-\frac{1}{(\tau_{\infty}^{(j)})^{2}+\beta}\right)$ and $\tau_{\infty}^{(j)}$ is a fixed point of the state that the AMP algorithm converges to. 
%Furthermore, the channel estimation error of user $n$ is denoted as $\Delta \mathbf{h}_{n}^{j}\triangleq\mathbf{h}_{n}^{j}-\hat{\mathbf{h}}_{n}^{j}$.

\subsection{Data Detection and Decoding}
In the second phase of the $j$-th transmission block, the AMP-based receiver first applies a low-complexity zero-forcing (ZF) equalizer to the received data signal as follows:
\begin{align}
\hat{\mathbf{W}}^{(j)}= \left[\{\hat{\mathbf{w}}_{k}^{\left(j\right)}\}_{k\in\hat{\mathcal{K}}^{\left(j\right)}}\right]=\hat{\mathbf{H}}^{(j)}\left((\hat{\mathbf{H}}^{(j)})^{\mathrm{H}}\hat{\mathbf{H}}^{(j)}\right)^{-1},
\end{align}

\noindent which effectively eliminates the inter-user interference with sufficient antennas at the BS, i.e., $M\geq |\hat{\mathcal{K}}^{(j)}|$. Therefore, the equalized data symbols of a user in $\hat{\mathcal{K}}^{(j)}$ can be written as follows:
\begin{align}
\begin{aligned}
\hat{\mathbf{s}}_{k}^{(j)}&=(\hat{\mathbf{w}}_{k}^{(j)})^{\mathrm{H}}\mathbf{Y}^{(j)}= (\hat{\mathbf{w}}_{k}^{(j)})^{\mathrm{H}} \left(\sum_{n \in \mathcal{K}^{(j)}} \mathbf{h}_{n}^{(j)} \mathbf{s}_{n}^{\mathrm{T}}+\mathbf{N}^{(j)}\right)= (\hat{\mathbf{w}}_{k}^{(j)})^{\mathrm{H}} \left(\sum_{n \in \mathcal{K}^{j}} \left(\hat{\mathbf{h}}_{n}^{(j)}+\Delta \mathbf{h}_{n}^{(j)}\right) \mathbf{s}_{n}^{\mathrm{T}}+\mathbf{N}^{(j)}\right)\\
&=\underbrace{(\hat{\mathbf{w}}_{k}^{(j)})^{\mathrm{H}} \hat{\mathbf{h}}_{k}^{(j)}\mathbf{s}_{k}^{\mathrm{T}}}_{\text{Signal}}+\underbrace{(\hat{\mathbf{w}}_{k}^{(j)})^{H}\sum_{n \in \mathcal{K}^{(j)},n \neq k} \hat{\mathbf{h}}_{n}^{(j)}\mathbf{s}_{n}^{\mathrm{T}}}_{\text{Inter-user interference}}+\underbrace{(\hat{\mathbf{w}}_{k}^{(j)})^{\mathrm{H}}\sum_{n \in \mathcal{K}^{(j)}} \Delta \mathbf{h}_{n}^{(j)}\mathbf{s}_{n}^{\mathrm{T}}}_{\text{Channel\ estimation\ error}}+\underbrace{(\hat{\mathbf{w}}_{k}^{(j)})^{\mathrm{H}}\mathbf{N}^{(j)}}_{\text{Noise}}, k\in\hat{\mathcal{K}}^{(j)},
\end{aligned}
\end{align}

%by using the $k$-th column of the above ZF equalizer $\hat{\mathbf{w}}_{k}^{j}$ and considering the channel estimation error,

\noindent where $\hat{\mathbf{H}}^{(j)}\triangleq \left[\{ {\hat{\mathbf{h}}_{i}^{(j)}}\}_{ i \in \hat{\mathcal{K}}^{(j)}} \right]$ and $\Delta \mathbf{h}_{n}^{(j)}\triangleq \mathbf{h}_{n}^{(j)}-\hat{\mathbf{h}}_{n}^{(j)}$ denotes the channel estimation error. The estimated codewords are passed to the channel decoder to recover the information bits.

%In order to evaluate the performance of all active users, it is further assumed that no data symbol can be correctly detected when $M<|\hat{\mathcal{K}}^{j}|$. This can be regarded as a conservative estimation of the system performance \cite{b16}.
However, although the AMP-based receiver can be recycled for multiple transmission blocks, it overlooks the critical user activity correlation among different transmission blocks, which can be exploited to improve the communication performance. In the next section, we develop a novel receiver for grant-free massive RA systems with retransmission.

%be fully utilized. Actually, many existing retranmission mechanisms have already considered it. For instance, hybrid ARQ (HARQ) uses the error correction technique called soft combining additionally, where the data packets that are not properly decoded are not discarded anymore, instead they will be combined with the next retransmission. A key observation in grant-free massive RA systems with retransmission is that the correlation with respect to the user activity and channel estimation in different transmission blocks can be considered in the receiver design, which is different from HARQ that utilizes the correlation from the perspective of data decoding. Therefore, in the next section, we endeavor to improve the performance of grant-free massive RA systems with retransmission by optimizing the receiver design, where some special characteristics of this system, e.g., the correlation of user activity between different transmission blocks and JADCE results in the previous transmission block, are utilized for retransmission blocks.

\section{Receiver Optimization for Retransmission Blocks} \label{sectioniv}
%Typically, in the grant-free RA system with a simple retransmission protocol, i.e., the stop-and-wait ARQ, the BS directly uses the same AMP-based receiver for JADCE for all transmission blocks, which neglects the correlation in different transmission blocks that can improve the performance. Actually, many existing retranmission mechanisms have already considered it. For instance, hybrid ARQ (HARQ) uses the error correction technique called soft combining additionally, where the data packets that are not properly decoded are not discarded anymore, instead they will be combined with the next retransmission.
In the retransmission blocks, the activity status of a user in $\mathcal{K}$ does not change unless its data has been correctly decoded in a previous transmission block. Therefore, the user activity correlation among different transmission blocks can be leveraged to optimize the AMP-based receiver. In this section, a new activity-correlation-aware receiver for the retransmission blocks is developed based on the correlated AMP algorithm, where the historical channel estimation results can be effectively utilized. In the sequel, we use \emph{correlated AMP-based receiver} interchangeably with \emph{activity-correlation-aware receiver}.

%which is different from HARQ that utilizes the correlation from the perspective of data decoding. Therefore, by jointly considering the correlation in user activity between the adjacent transmission blocks (i.e., $(j-1)$-th and $j$-th) and JADCE results in the previous transmission block (i.e., $(j-1)$-th), the receiver for retransmission blocks (i.e., $j>1$) can be designed in a different manner.
\subsection{Correlated AMP for JADCE} \label{IV.A}
We first analyze the activity pattern of a user in adjacent transmission blocks, i.e., the $(j-1)$-th and $j$-th transmission block ($j = 2, \cdots, J$), which belongs to one of the following three cases:

\textbf{Case 1}: A user $n$ does not transmit in both transmission blocks, i.e., $u_{n}^{(j-1)}=0$ and $u_{n}^{(j)}=0$. This indicates that the user is either inactive or its data has been successfully decoded in one of the first $j-2$ transmission blocks.

%\textit{Case 2}: An user is not transmitted in the $(j-1)$-th transmission block but in the $j$-th block, i.e., $u_{n}^{j-1}=0$ and $u_{n}^{j}=1$. This shows that the user is scheduled to transmit in the $j$-th transmission block.

\textbf{Case 2}: A user $n$ transmits in the $(j-1)$-th transmission block but not in the $j$-th block, i.e., $u_{n}^{(j-1)}=1$ and $u_{n}^{(j)}=0$, which means it is an active user and its data is decoded successfully in the $(j-1)$-th transmission block.

\textbf{Case 3}: A user $n$ transmits in both transmission blocks, i.e., $u_{n}^{(j-1)}=1$ and $u_{n}^{(j)}=1$, which means it is an active user but its data has not been decoded successfully in the first $j-1$ transmission blocks.

Since the statistical channel inversion power control is applied, given the numbers of active users in the first $j$ transmission blocks, the above three cases occur with probability $\epsilon_1^{(j)}=\frac{N-K^{(j-1)}}{N}$, $\epsilon_2^{(j)}=\frac{K^{(j-1)}-K^{(j)}}{N}$, and $\epsilon_3^{(j)}=\frac{K^{(j)}}{N}$, respectively, which provides additional prior information for JADCE in the AMP-based receiver. However, it is generally hard to incorporate such user activity correlation into the conventional AMP algorithm due to its memoryless nature. Fortunately, the AMP framework with side information (AMP-SI) developed in \cite{annama2019} provides a promising framework, which integrates the historical information (HI) into the denoiser of the AMP algorithm and thus endows it with memory. Driven by this framework, we propose the correlated AMP algorithm for JADCE in the retransmission blocks, where the HI is judiciously designed based on the user activity correlation and channel estimation results in the previous transmission block.

%Case $i$ occurs with probability $\epsilon_i^{j}$ and thus $\sum_{i=1}^{3} \epsilon_{i}^{j}=1$. Here, $\epsilon_1^{j}=\frac{N-K^{j-1}}{N}$, $\epsilon_2^{j}=\frac{K^{j-1}-K^{j}}{N}$, and $\epsilon_3^{j}=\frac{K^{j}}{N}$. Under the modeled user activity correlation, we should not detect the user activity and estimate their channel coefficients over consecutive transmission blocks in an independent manner, since the transition probabilities of the

% Apart from the correlation in user activity mentioned above, we also endeavor to incorporate the channel estimation results in the $(j-1)$-th transmission block to design the receiver for the $j$-th transmission block. This is because when Bayesian signal recovery algorithms have access to more precise prior information, it can offer the potential to markedly improve recovery quality, and the channel estimation result in the previous transmission block is a supplement to the prior information.

We consider the non-trivial case where the data packets of at least one active users in $\mathcal{K}$ have not been successfully decoded in the $j$-th ($j\geq 2$) transmission block, i.e., $K^{(j)}\neq 0$. The correlated AMP algorithm starts from $\mathbf{H}_{0}^{(j)}=\mathbf{0}$ and $\mathbf{R}_{0}^{(j)}=\mathbf{Y}_{p}^{(j)}$, and iterates as follows:
\begin{align}
    \mathbf{h}_{n, t+1}^{(j)}=\eta_{t}^{j}\left(\mathbf{h}_{n, t}^{(j)}+\left(\mathbf{R}_{t}^{(j)}\right)^{\mathrm{H}} \mathbf{p}_{n}, \tilde{\mathbf{h}}_{n}^{(j-1)}\right),
\label{esthsi}
\end{align}
\begin{align}
\mathbf{R}_{t+1}^{(j)} =\mathbf{Y}_{p}^{(j)}-\mathbf{P} \mathbf{H}_{t+1}^{(j)}+\frac{N}{L} \mathbf{R}_{t}^{(j)}\sum_{n=1}^{N}\frac{\eta_{t}^{\prime (j)}\left(\mathbf{h}_{n, t}^{(j)}+\left(\mathbf{R}_{t}^{(j)}\right)^{\mathrm{H}} \mathbf{p}_{n}, \tilde{\mathbf{h}}_{n}^{(j-1)}\right)}{N},
\end{align}

\noindent where $\tilde{\mathbf{h}}_{n}^{(j-1)}$ denotes the HI derived from the channel estimation results of the AMP algorithm executed in the first transmission block (for $j=2$) or the correlated AMP algorithm executed in the $(j-1)$-th transmission block (for $j \geq 3$). Besides, $\eta_{t}^{(j)}(\cdot,\cdot)$ denotes the denoiser, and $\eta_{t}^{\prime (j)}(\cdot,\cdot)$ is the ﬁrst-order derivative of $\eta_{t}^{(j)}(\cdot,\cdot)$. In contrast to AMP, the denoiser of correlated AMP applies to both the state and HI.

\subsection{State Evolution}
The state evolution of the correlated AMP algorithm bears similar properties as the AMP algorithm, which can be expressed in the following lemma. 

\begin{lemma} \label{lemma2}
Define two random variables as $\mathbf{X}^{(j)}$ and $\mathbf{V}^{(j)}$ with distributions $(1-\lambda^{(j)}) \delta_{0}+\lambda^{(j)} \mathcal{CN}(0,\beta\mathbf{I})$ and $\mathcal{CN}(0,\mathbf{I})$, respectively. In the asymptotic regime where $N,L,K^{(j)} \rightarrow \infty$, the state evolution of the correlated AMP algorithm is expressed as follows:
\begin{align}
\mathbf{\Sigma}^{(j)}_{t+1}=\frac{\sigma^{2}}{L} \mathbf{I}+\frac{N}{L} \mathbb{E}_{\mathbf{X}^{(j)},\mathbf{V}^{(j)}}\left[\Big|\Big|\eta^{(j)}_{t}(\mathbf{X}^{(j)}+(\mathbf{\Sigma}^{(j)}_{t})^{\frac{1}{2}} \mathbf{V}^{(j)},\tilde{\mathbf{X}}^{(j-1)})-\mathbf{X}^{(j)}\Big|\Big|_{2}^{2}\right],
\end{align}

\noindent where $\tilde{\mathbf{X}}^{(j-1)}$ is a random variable that captures the distribution of $\tilde{\mathbf{h}}_{n}^{(j-1)}$. It can be expressed as $\tilde{\mathbf{X}}^{(j-1)}=\mathbf{X}^{(j-1)}+(\mathbf{\Sigma}^{(j-1)}_{\infty})^{\frac{1}{2}}\mathbf{V}^{(j-1)}$ with $\mathbf{X}^{(j-1)}\sim (1-\lambda^{(j-1)}) \delta_{0}+\lambda^{(j-1)} \mathcal{CN}(0,\beta\mathbf{I})$ and $\mathbf{V}^{(j-1)} \sim \mathcal{CN}(0,\mathbf{I})$. Besides, $\mathbf{\Sigma}^{(j)}_{t}$ denotes the state matrix in the $t$-th iteration of the correlated AMP algorithm. When the MMSE denoiser is applied, the state matrix of the correlated AMP algorithm stays as a scaled identity matrix, i.e., $\mathbf{\Sigma}^{(j)}_{t}=(\tau^{(j)}_{t})^{2}\mathbf{I}$.
\end{lemma}
\begin{proof}
This lemma can be proved by following similar lines in Appendix B of \cite{lliu2018}, with the main difference of replacing the MMSE denoiser $\eta_{t}^{(j)}(\cdot)$ in the AMP algorithm with $\eta_{t}^{(j)}(\cdot,\cdot)$ in the correlated AMP algorithm.
\end{proof}

The main difference in the state evolution between the conventional AMP and correlated AMP algorithm lies on the denoiser, since correlated AMP takes the HI as additional input. In the following subsection, we elaborate the designs of the HI and the HI-assisted denoiser for the correlated AMP algorithm.

\subsection{Historical Information and Denoiser Design}
%Based on the above correlated AMP framework, we aim at utilizing the HI obtained from channel estimation results in the previous transmission block, and thus the key problem is how to transform the channel estimation results in the $(j-1)$-th transmission block, to the applicable HI for AMP-based channel estimation in the $j$-th transmission block. 
We propose to derive the HI $\tilde{\mathbf{h}}_{n}^{(j-1)}$ that connects $\hat{\mathbf{h}}_{n}^{(j-1)}$ with the actual effective channel coefficients $\mathbf{h}_{n}^{(j)}$ in the $j$-th transmission block. As shown in Fig. \ref{HI}, since the user activity correlation between adjacent transmission blocks has been obtained in Section \ref{IV.A}, the dependency of the actual effective channel coefficients in adjacent transmission blocks, i.e., $\mathbf{h}_{n}^{(j-1)}$ and $\mathbf{h}_{n}^{(j)}$, can be established accordingly. Hence, it remains to link $\hat{\mathbf{h}}_{n}^{(j-1)}$ with $\mathbf{h}_{n}^{(j-1)}$. For this purpose, the state evolution of the AMP/correlated AMP algorithm, which tracks the mean square estimation error, provides the statistical hint. 
\begin{figure}[htbp]
\centering
\includegraphics[width=5in]{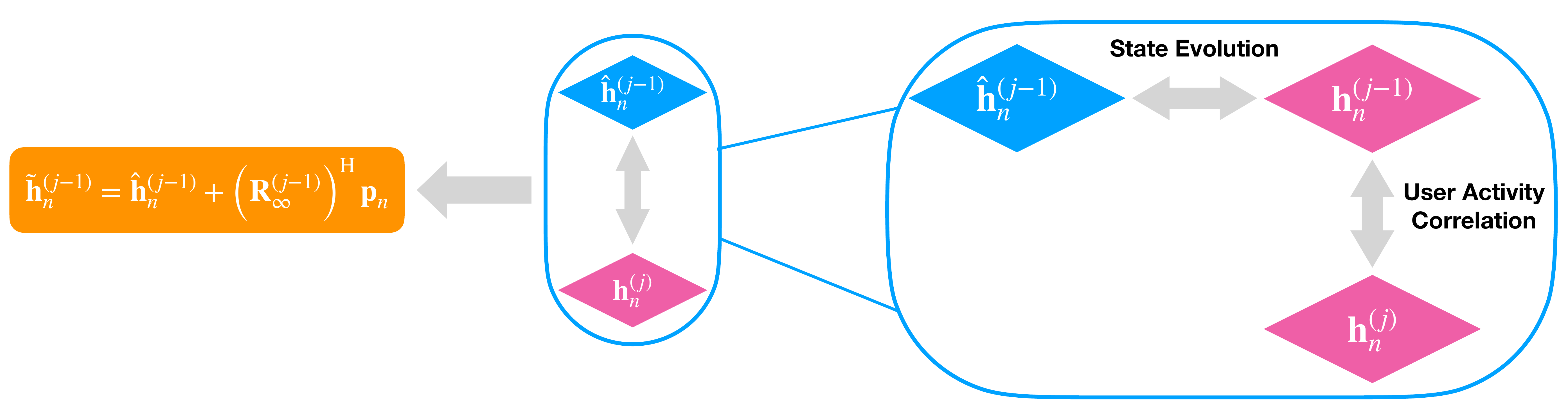}
\caption{Relationship of different variables to facilitate the historical information (HI) design.}
\vspace{-0.3cm}
\label{HI}
\end{figure}

Specifically, the state evolution in each iteration of the AMP/correlated AMP algorithm indicates that in the statistical sense, applying the denoiser $\eta^{(j-1)}_{t}(\cdot)$/$\eta^{(j-1)}_{t}(\cdot,\cdot)$ to $\mathbf{h}^{(j-1)}_{n,t}+(\mathbf{R}^{(j-1)}_{t})^{\mathrm{H}}\mathbf{p}_{n}$ (see (\ref{esth})/(\ref{esthsi})) is equivalent to applying it to $\mathbf{h}^{(j-1)}_{n}+(\mathbf{\Sigma}^{(j-1)}_{t})^{\frac{1}{2}}\mathbf{v}^{(j-1)}_{n}$ \cite{annama2019,mba2011}, where $\mathbf{v}^{(j-1)}_{n}\sim \mathcal{CN}(\mathbf{0},\mathbf{I})$ is a noise term independent of $\mathbf{h}_{n}^{(j-1)}$. Since $\hat{\mathbf{h}}_{n}^{(j-1)}=\mathbf{h}_{n,\infty}^{(j-1)}$, we can establish the statistical relationship between $\hat{\mathbf{h}}_{n}^{(j-1)}$ and $\mathbf{h}_{n}^{(j-1)}$ as $\hat{\mathbf{h}}_{n}^{(j-1)}=\mathbf{h}^{(j-1)}_{n}+(\mathbf{\Sigma}^{(j-1)}_{\infty})^{\frac{1}{2}}\mathbf{v}^{(j-1)}_{n}$, where $\mathbf{\Sigma}^{(j-1)}_{\infty}$ denotes the state of the AMP/correlated AMP algorithm upon it converges. Therefore, we design the HI for the correlated AMP algorithm as follows:
\begin{align}
\tilde{\mathbf{h}}_{n}^{(j-1)}=\hat{\mathbf{h}}_{n}^{(j-1)}+\left(\mathbf{R}_{\infty}^{(j-1)}\right)^{\mathrm{H}}\mathbf{p}_{n}.
\label{SI}
\end{align}

Based on the proposed HI in (\ref{SI}), we next design the denoiser $ \eta_{t}^{(j)}(\cdot,\cdot)$. We still adopt the MMSE denoiser for its superior performance and elegant theoretical properties. Define two auxiliary terms $\mathbf{a}_{n,t}^{(j)} \triangleq \mathbf{h}^{(j)}_{n}+\tau^{(j)}_{t}\mathbf{v}^{(j)}_{n}$ and $\mathbf{b}_{n,\infty}^{(j-1)}\triangleq \mathbf{h}^{(j-1)}_{n}+\tau^{(j-1)}_{\infty}\mathbf{v}^{(j-1)}_{n}$, where the distributions of $\mathbf{a}_{n,t}^{(j)}$, $\mathbf{b}_{n,\infty}^{(j-1)}$, $\mathbf{v}^{(j)}_{n}$, and $\mathbf{v}^{(j-1)}_{n}$ are respectively captured by random vectors $\mathbf{X}^{(j)}+(\mathbf{\Sigma}^{(j)}_{t})^{\frac{1}{2}}\mathbf{V}^{(j)}$, $\tilde{\mathbf{X}}^{(j-1)}$, $\mathbf{V}^{(j)}$, and $\mathbf{V}^{(j-1)}$. The MMSE denoiser can be obtained in the following theorem.

\begin{theorem} \label{theorem1}
The MMSE denoiser of the correlated AMP algorithm is given by the conditional expectation $\mathbb{E}[\mathbf{X}^{(j)}|\mathbf{X}^{(j)}+(\mathbf{\Sigma}^{(j)}_{t})^{\frac{1}{2}}\mathbf{V}^{(j)}=\mathbf{a}_{n,t}^{(j)}, \tilde{\mathbf{X}}^{(j-1)}=\mathbf{b}_{n,\infty}^{(j-1)}]$ as follows:
\begin{align}
    \eta_{t}^{(j)}(\mathbf{a}_{n,t}^{(j)},\mathbf{b}_{n,\infty}^{(j-1)})=\frac{\frac{\beta\mathbf{a}_{n,t}^{(j)}}{\beta+(\tau_t^{(j)})^2}}{1+\Phi_{n,t,1}^{(j)}\times \frac{\epsilon_2^{(j)}+\epsilon_1^{(j)}\Phi_{n,2}^{(j-1)}}{\epsilon_3^{(j)}\Phi_{n,2}^{(j-1)}}},
\label{denoiserampsi}
\end{align}
\noindent where 
\vspace{-0.3cm}
\begin{align}
    \Phi_{n,t,1}^{(j)}=\left(\frac{\beta+(\tau_t^{(j)})^2}{(\tau_t^{(j)})^2}\right)^M e^{\left(\frac{1}{\beta+(\tau_t^{(j)})^2}-\frac{1}{(\tau_t^{(j)})^2}\right)||\mathbf{a}_{n,t}^{(j)}||^2},
\end{align}
\vspace{-0.3cm}
\begin{align}
    \Phi_{n,2}^{(j-1)}=\left(\frac{\beta+(\tau_{\infty}^{(j-1)})^2}{(\tau_{\infty}^{(j-1)})^2}\right)^M e^ {\left(\frac{1}{\beta+(\tau_{\infty}^{(j-1)})^2}-\frac{1}{(\tau_{\infty}^{(j-1)})^2}\right)||\mathbf{b}_{n,\infty}^{(j-1)}||^2}.
\end{align}
\end{theorem}
\begin{proof} Please refer to Appendix \ref{appendixa}.

\end{proof}
\vspace{-0.3cm}
In the above MMSE denoiser, we make use of both the user activity correlation between adjacent transmission blocks and the channel estimation results obtained in the previous transmission block. Notice that the term $\frac{\epsilon_2^{(j)}+\epsilon_1^{(j)}\Phi_{n,2}^{(j-1)}}{\epsilon_3^{(j)}\Phi_{n,2}^{(j-1)}}$ is parameterized by the user activity correlation. If $\epsilon_1^{(j)}=\frac{N-K^{(j)}}{N}$, $\epsilon_2^{(j)} = 0 $, and $\epsilon_3^{(j)} = \frac{K^{(j)}}{N}$, the user activity in the $(j-1)$-th transmission block is not considered when estimating the effective channel coefficients in the $j$-th transmission block. This can be interpreted as merging Case 1 and Case 2 presented in Section \ref{IV.A} into a new case that a user $n$ does not transmit in the $j$-th transmission block ($u_{n}^{(j)}=0$), while Case 3 implies a user $n$ transmits in the $j$-th transmission block ($u_{n}^{(j)}=1$). In other words, the MMSE denoiser of the correlated AMP algorithm degenerates to that of the AMP algorithm.

\subsection{User Activity Detection}
In the first transmission block, upon the AMP algorithm converges, the active user set can be obtained from (\ref{activeamp}). In the retransmission blocks, the set of active users is obtained according to the following theorem.

\begin{theorem} \label{theorem2}
When the correlated AMP algorithm converges, the set of active users in the $j$-th ($j>1$) transmission block, i.e., $\hat{\mathcal{K}}^{(j)}$, is estimated based on the log-likelihood ratio test as follows:
\begin{align}
\hat{u}^{(j)}_{n}=\left\{\begin{array}{l}
1, \text { if }\lVert\mathbf{h}^{(j)}_{n,\infty}+\left(\mathbf{R}^{(j)}_{\infty}\right)^{\mathrm{H}}\mathbf{p}_{n}\rVert^{2}\geq l_{r}^{(j)} \\
0, \text { if }\lVert\mathbf{h}^{(j)}_{n,\infty}+\left(\mathbf{R}^{(j)}_{\infty}\right)^{\mathrm{H}}\mathbf{p}_{n}\rVert^{2}< l_{r}^{(j)}
\end{array}\right., n \in \mathcal{N},
\label{activeampsi}
\end{align}

\noindent where
\begin{align}
l_{r}^{(j)}\triangleq \frac{M \ln \left(1+\frac{\beta}{(\tau^{(j)}_{\infty})^{2}}\right)+\ln\left(\frac{\frac{\epsilon_{1}^{(j)}}{\epsilon_{1}^{(j)}+\epsilon_{2}^{(j)}}\Phi_{2}^{(j-1)}}{1-\frac{\epsilon_{2}^{(j)}}{\epsilon_{1}^{(j)}+\epsilon_{2}^{(j)}}\Phi_{2}^{(j-1)}}\right)}{\frac{1}{(\tau_{\infty}^{(j)})^2}-\frac{1}{(\tau_{\infty}^{(j)})^{2}+\beta}}
\end{align}

\noindent denotes the detection threshold, and $\Phi_{2}^{(j-1)}=\left(\frac{\beta+(\tau_{\infty}^{(j-1)})^2}{(\tau_{\infty}^{(j-1)})^2}\right)^M e^ {\left(\frac{1}{\beta+(\tau_{\infty}^{(j-1)})^2}-\frac{1}{(\tau_{\infty}^{(j-1)})^2}\right)\left(\beta\lambda^{(j-1)}+(\tau_{\infty}^{(j-1)})^{2}\right)}$.
\end{theorem}

\begin{proof} Please refer to Appendix \ref{appendixb}.

\end{proof}
\vspace{-0.3cm}
Once the JADCE results are obtained, the information bits of the set of users that are estimated as active can be detected and decoded as in the conventional AMP-based receiver. The performance of the correlated AMP-based receiver is analyzed in the following section. It is noteworthy that while we assume the active users simply retransmit the same coded bits in different transmission blocks in this paper, the proposed receiver can be easily extended when chase combining or incremental redundancy is activated in the retransmission blocks \cite{fjab2021}. The detailed procedures of the activity-correlation-aware receiver are summarized in Fig. \ref{rec}.
\begin{figure}[htbp]
\centering
\includegraphics[width=4.6in]{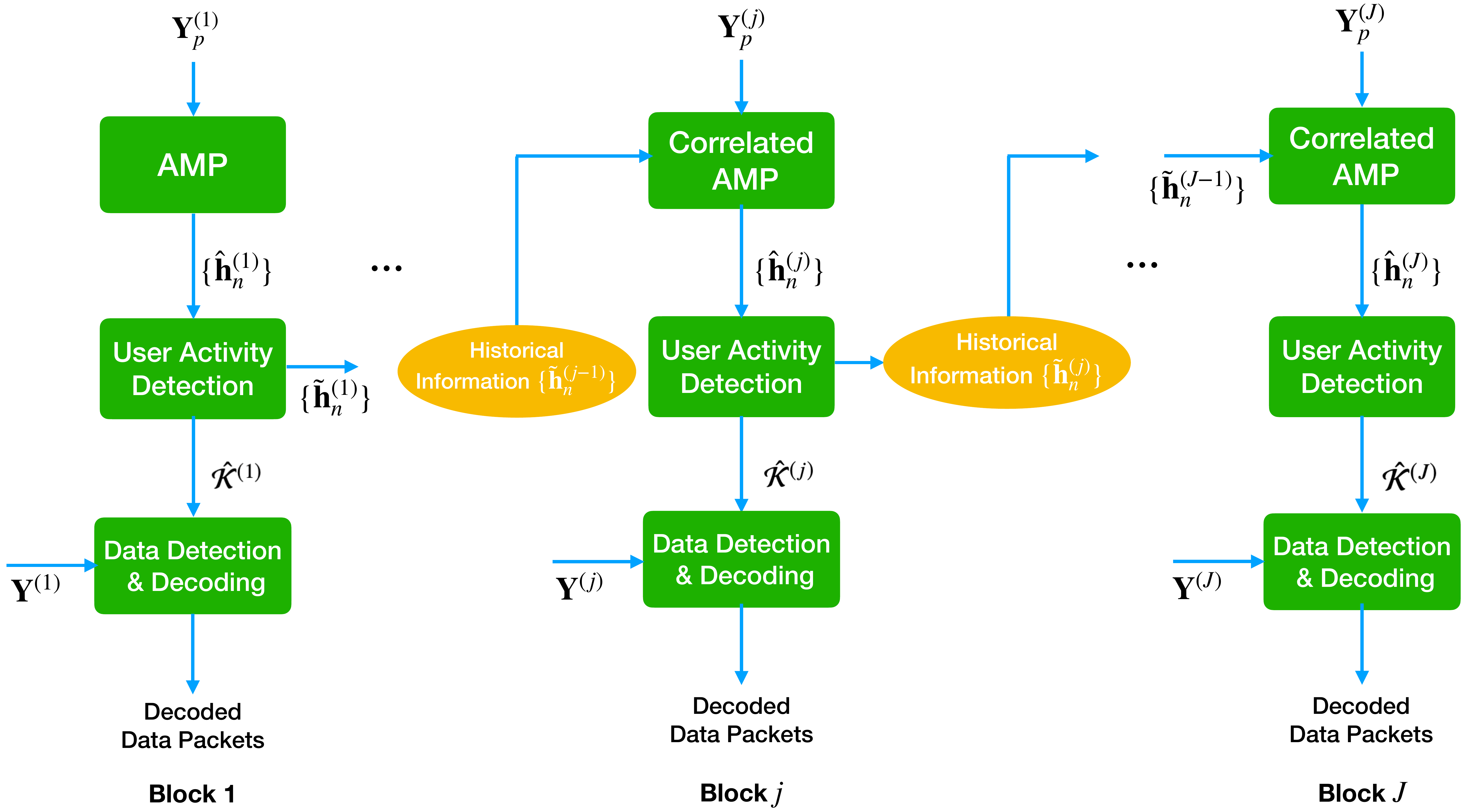}
\caption{Detailed procedures of the activity-correlation-aware receiver, where the correlated AMP algorithm is adopted for JADCE except in the first transmission block.}
\label{rec}
\end{figure}
\vspace{-0.3cm}
\section{Performance Analysis \label{sectionv}}
In this section, we first analyze the performance of the correlated AMP-based receiver in a single transmission block, including the activity detection and channel estimation error, and the BLER. Based on the analysis, we then derive the FER of grant-free massive RA systems with retransmission. Notice that by the design of the HI, the performance of the correlated AMP-based receiver in the $j$-th transmission block ($j>1$) is conditioned on the number of active users in the $(j-1)$-th and $j$-th transmission block, i.e., it should be parameterized by $K^{\left(j-1\right)}$ and $K^{\left(j\right)}$. For brevity, we omit $K^{\left(j-1\right)}$ and $K^{\left(j\right)}$ in the following derivations unless otherwise specified. 

\subsection{Activity Detection and Channel Estimation Error}
We focus on the regime with $L\geq K$ where the AMP-based algorithms were reported to have a stable behavior \cite{lliu2018}. In the $j$-th transmission block, when the AMP/correlated AMP algorithm converges, the performance of user activity detection can be characterized via the fixed point state $\tau_{\infty}^{j}$, as expressed in the following lemma.

\begin{lemma} \label{lemma3}
In the $j$-th transmission block, the probability of missed detection, defined as $P_{M,k}^{(j)}\triangleq P(\hat{u}_{k}^{(j)}=0|u_{k}^{(j)}=1)$, and the probability of false alarm, defined as $P_{F,n}^{(j)}\triangleq P(\hat{u}_{n}^{(j)}=1|u_{n}^{(j)}=0)$, can be expressed as follows:
\begin{align}
\label{md}
P^{(j)}_{M,k}\!=\!P^{(j)}_{M}\!\triangleq\!\left\{
\begin{aligned}
&\frac{1}{\Gamma(M)} \bar{\gamma}\left(M, \frac{M(\tau_{\infty}^{(1)})^{2}}{\beta}\ln\left(1+\frac{\beta}{(\tau_{\infty}^{(1)})^{2}}\right)\right), j=1,\\
&\frac{1}{\Gamma(M)} \bar{\gamma}\left(M, \frac{(\tau^{(j)}_{\infty})^{2}}{\beta}\left(M\ln\left(\!1\!+\!\frac{\beta}{(\tau^{(j)}_{\infty})^{2}}\right)\!+\!\ln\left(\frac{\epsilon_{1}^{(j)}\Phi_{2}^{(j-1)}}{\!\epsilon_{1}^{(j)}\!+\!\epsilon_{2}^{(j)}\!-\!\epsilon_{2}^{(j)}\Phi_{2}^{(j-1)}}\right)\right)\right), j>1, \notag
\end{aligned}
\right.
\\k\in\mathcal{K}^{(j)},
\end{align}
\begin{align}
\label{fa}
P^{(j)}_{F,n}\!=\! P^{(j)}_{F}\! \triangleq\! \left\{
\begin{aligned}
&1\!-\!\frac{1}{\Gamma(M)} \bar{\gamma}\left(M, M\left(1+ \frac{(\tau_{\infty}^{(1)})^{2}}{\beta}\right)\ln\left(1+\frac{\beta}{(\tau_{\infty}^{(1)})^{2}}\right)\right), j=1,\\
&1\!-\!\frac{1}{\Gamma(M)} \bar{\gamma}\left(\!M,\!\left(1\!+\!\frac{(\tau^{(j)}_{\infty})^{2}}{\beta}\right)\!\left(\!M\!\ln\left(\!1\!+\!\frac{\beta}{(\tau^{(j)}_{\infty})^{2}}\right)\!+\!\ln\!\left(\frac{\epsilon_{1}^{(j)}\Phi_{2}^{(j-1)}}{\epsilon_{1}^{(j)}\!+\!\epsilon_{2}^{(j)}\!-\!\epsilon_{2}^{(j)}\Phi_{2}^{(j-1)}}\right)\right)\right), j>1, \notag
\end{aligned}
\right.
\\n\in \mathcal{N}\setminus \mathcal{K}^{(j)},
\end{align}

\noindent where $\Gamma(\cdot)$ and $\bar{\gamma}(\cdot)$ denote the Gamma function and the upper incomplete Gamma function, respectively.
\end{lemma}

\begin{proof}
Please refer to Appendix \ref{appendixc}.
\end{proof}

Besides, the state evolution analysis of the AMP/correlated AMP algorithm facilitates the evaluation of the channel estimation performance in the asymptotic regime, which is formalized in the following lemma. 

\begin{lemma} \label{lemma4}
Assume the number of users $N$, the number of active users in the $j$-th transmission block $K^{(j)}$, the number of BS antennas $M$, and the length of the pilot sequences $L$ all approach inﬁnity. When the AMP/correlated AMP algorithm converges, the channel estimation error $\Delta \mathbf{h}_{k}^{(j)}$ of an active user in the $j$-th transmission block is zero mean and the covariance matrix is given as follows:
\vspace{-0.3cm}
\begin{align}
\operatorname{Cov}\left(\Delta \mathbf{h}_{k}^{(j)}, \Delta \mathbf{h}_{k}^{(j)}\right)=\Delta v_{k}^{(j)} \mathbf{I}, k \in \mathcal{K}^{(j)},
\label{channelerror}
\end{align}

\noindent where $\Delta v_{k}^{(j)}=\frac{\beta (\tau_{\infty}^{(j)})^{2}}{\beta+(\tau_{\infty}^{(j)})^{2}}$ for $k \in \mathcal{K}^{(j)}\cap \hat{\mathcal{K}}^{(j)}$ and $\Delta v_{k}^{(j)}=\beta$ for $k \in \mathcal{K}^{(j)} \setminus \hat{\mathcal{K}}^{(j)}$.
\end{lemma}  

\begin{proof}
This lemma can be proved by following similar lines in Appendices D and G of \cite{lliu2018}, with the main difference of replacing the MMSE denoiser $\eta_{t}^{(j)}(\cdot)$ in the AMP algorithm with $\eta_{t}^{(j)}(\cdot,\cdot)$ in the correlated AMP algorithm.
\end{proof}

In Lemma \ref{lemma4}, there is an underlying assumption that the number of active users in a transmission block approaches infinity. Although this is an idealized assumption that cannot hold in practice since the number of active users is likely to decrease over the transmission blocks, we show empirically in Section \ref{sectionvi} that (\ref{channelerror}) is accurate in practical system settings.

\subsection{Block Error Rate Analysis} \label{sectionvb}
%Based on the AMP-based (AMP/correlated AMP) algorithm for JADCE, the user activity detection and channel estimation errors can be characterized from the state evolution. Then, in order to develop the data error rate, the post-processing SNR distribution according to (9) is very crucial, which is conditioned on the outcome of user activity detection. After obtaining the conditional distribution of post-processing SNR via the random matrix theory, the data error rate is derived by accounting for all conditions. In addition, we further simplify our analysis by making some approximations and a close-form solution is presented.

There are two major sources that contribute to the data error of an active user in each transmission block: First, if an active user is not detected, no data detection and decoding is performed. Also, even though an active user is detected, data detection and decoding may still be erroneous due to channel and noise distortions. By taking into account their composite effects, the BLER of an active user in the $j$-th transmission block can be expressed as follows:
\begin{align}
    P_{e,k}^{(j)}&=P_{M,k}^{(j)}+\left(1-P_{M,k}^{(j)}\right)\bar{\varepsilon}_{k}^{(j)}, k\in\mathcal{K}^{(j)},
\label{bler}
\end{align}

\noindent where $\bar{\varepsilon}_{k}^{(j)}$ is the average BLER of an active user $k$ given it is detected. In the finite blocklength regime, error-free transmission is unachievable. Nevertheless, the relationship between the data error and channel coding rate can be tightly approximated via the following lemma.

\begin{lemma} \label{lemma5}
\cite{b17} In the finite blocklength regime, the maximum achievable rate of an additive white Gaussian noise (AWGN) channel can be tightly approximated via the following normal approximation: \vspace{-0.3cm}
\begin{align}
    R(d,\varepsilon)\approx C(\gamma)-\sqrt{\frac{V(\gamma)}{d}}Q^{-1}(\varepsilon),
\end{align}

\noindent where $d$ is the blocklength, $\varepsilon$ is the error probability, and $\gamma$ is the SNR. The term $C(\gamma)\triangleq\log_{2}{(1+\gamma)}$ is the channel capacity and $V(\gamma)\triangleq\frac{\gamma(\gamma+2)}{2(\gamma+1)^2}\log_{2}^2{e}$ is the channel dispersion. Besides, $Q^{-1}(x)$ denotes the inverse function of $Q(x)\triangleq \int_{x}^{\infty} \frac{1}{\sqrt{2 \pi}} e^{-t^{2} / 2} dt$. 
\end{lemma}

Accordingly, the BLER of an active user that is correctly detected in the $j$-th transmission block can be approximated as follows: \vspace{-0.3cm}
\begin{align}
   \varepsilon_{k}^{(j)}\left(\gamma_{k}^{(j)}\right) \approx Q\left(\frac{C(\gamma_{k}^{(j)})-R}{\sqrt{V(\gamma_{k}^{(j)})/d}}\right), k \in \mathcal{K}^{(j)}\cap \hat{\mathcal{K}}^{(j)},
\label{fblerror}
\end{align}
\noindent where $\gamma_{k}^{(j)}$ is the post-processing SNR. It can be derived based on (9) as follows:
\begin{align}
\begin{aligned}
    \gamma_{k}^{(j)}&=\frac{\mathbb{E}\left[|| (\hat{\mathbf{w}}_{k}^{(j)})^{\mathrm{H}} \hat{\mathbf{h}}_{k}^{(j)}(\mathbf{s}_{k}^{(j)})^{\mathrm{T}}||^2\right]}{\mathbb{E}\Big[||(\hat{\mathbf{w}}_{k}^{(j)})^{\mathrm{H}}\sum_{n \in \mathcal{K}^{(j)},n \neq k}  \hat{\mathbf{h}}_{n}^{(j)}(\mathbf{s}_{n}^{(j)})^{\mathrm{T}}||^2+||(\hat{\mathbf{w}}_{k}^{(j)})^{\mathrm{H}}\sum_{n \in \mathcal{K}^{(j)}}  \Delta \mathbf{h}_{n}^{(j)}(\mathbf{s}_{n}^{(j)})^{\mathrm{T}}||^2+||(\hat{\mathbf{w}}_{k}^{(j)})^{\mathrm{H}}\mathbf{N}^{(j)}||^2\Big]}\\
    &\overset{(\text{a})}{=}\frac{d}{0+||\hat{\mathbf{w}}_{k}^{(j)}||^2\left(\sum_{n \in \mathcal{K}^{(j)}} \mathbb{E}\left[||\Delta \mathbf{h}_{n}^{(j)}||^2\right]d+\mathbb{E}\left[||\mathbf{N}^{(j)}||^2\right]\right)}\\
    &\overset{(\text{b})}{=}\frac{1}{\left[\left((\hat{\mathbf{H}}^{(j)})^{\mathrm{H}}\hat{\mathbf{H}}^{(j)}\right)^{-1}\right]_{kk} \left(\sum_{n \in \mathcal{K}^{(j)}} \operatorname{Var}(\Delta \mathbf{h}_{n}^{(j)})+\sigma^2\right)},
\end{aligned}   
\end{align}

\noindent where (a) is because $\mathbb{E}\left[(\mathbf{s}_{n}^{j})^{\mathrm{T}}\mathbf{s}_{n}^{j}\right]=d$ and $(\hat{\mathbf{H}}^{(j)})^{\mathrm{H}}\hat{\mathbf{W}}^{(j)}=\mathbf{I}$, and (b) follows $\Delta \mathbf{h}_{n}^{(j)} \sim \mathcal{CN}(\mathbf{0},\Delta v_{n}^{(j)}\mathbf{I})$, $ n\in\mathcal{K}^{(j)}$. 

To obtain the average BLER, it is critical to determine the distribution of the post-processing SNR, which is conditioned on the outcome of user activity detection. In particular, the dimension of matrix $(\hat{\mathbf{H}}^{(j)})^{\mathrm{H}}\hat{\mathbf{H}}^{(j)}$ depends on the number of users that are estimated as active, which is influential on the distribution of $[((\hat{\mathbf{H}}^{(j)})^{\mathrm{H}}\hat{\mathbf{H}}^{(j)})^{-1}]_{kk}$. In addition, the term $\sum_{n \in \mathcal{K}^{(j)}}\operatorname{Var}(\Delta \mathbf{h}_{n}^{(j)})$ denotes the cumulative channel estimation error of all the active users, which depends on whether they are detected correctly. In the following lemma, we elaborate the distribution of the post-processing SNR conditioned on the number of missed detection and false alarm users.

\begin{lemma} \label{lemma6}
Suppose there are $e^{(j)}$ ($e^{(j)}=0,1,\cdots,K^{(j)}$) missed detection and $f^{(j)}$ ($f^{(j)}=0,1,\cdots,N-K^{(j)}$) false alarm users in the $j$-th transmission block. If $M\geq|\hat{\mathcal{K}}^{(j)}|= K^{(j)}-e^{(j)}+f^{(j)}$, the post-processing SNR of a correctly detected active user, i.e., $k \in \mathcal{K}^{(j)}\cap \hat{\mathcal{K}}^{(j)}$, can be expressed as follows:
\begin{align}
    \gamma_{k|e^{(j)}f^{(j)}}^{(j)}=\frac{1}{\left[\left((\hat{\mathbf{H}}^{(j)})^{\mathrm{H}}\hat{\mathbf{H}}^{(j)}\right)^{-1}\right]_{kk} \left((K^{(j)}-e^{(j)})\frac{\beta(\tau_{\infty}^{(j)})^2}{\beta+(\tau_{\infty}^{(j)})^2}+e^{(j)} \beta+\sigma^2\right)},
\end{align}

\noindent which is a Chi-square random variable with the probability density function (PDF) given as $\Psi_{\gamma_{k|e^{(j)}f^{(j)}}^{(j)}}(x)=\frac{x^{\theta_{2}^{(j)}-1}e^{-x/(2\theta_{1}^{(j)})}}{2^{\theta_{2}^{(j)}}\Gamma(\theta_{2}^{(j)})}$ $\left(x\geq 0\right)$, and $\theta_{1}^{(j)}=\frac{\beta^2}{(K^{(j)}-e^{(j)})\beta(\tau_{\infty}^{(j)})^2+(e^{(j)} \beta+\sigma^2)(\beta+(\tau_{\infty}^{(j)})^2)}$, $\theta_{2}^{(j)}=M-|\hat{\mathcal{K}}^{(j)}|+1$. 
\end{lemma}

\begin{proof}
Since the columns in $\hat{\mathbf{H}}^{(j)}$ are linearly independent, according to the random matrix theory \cite{b18}, $1/[((\hat{\mathbf{H}}^{(j)})^{\mathrm{H}}\hat{\mathbf{H}}^{(j)})^{-1}]_{kk}$ can be rewritten as $(\hat{\mathbf{h}}_{k}^{(j)})^{\mathrm{H}}\tilde{\mathbf{M}}^{(j)}\hat{\mathbf{h}}_{k}^{(j)}$, where $\tilde{\mathbf{M}}^{(j)}$ is a non-negative Hermitian matrix with $|\hat{\mathcal{K}}^{(j)}|-1$ eigenvalues equal zero and $M-|\hat{\mathcal{K}}^{(j)}|+1$ eigenvalues equal 1. Since $\hat{\mathbf{h}}_{k}^{(j)} \sim \mathcal{CN}(0,\frac{\beta^2}{\beta+(\tau_{\infty}^{(j)})^2}\mathbf{I})$, $\forall k \in \mathcal{K}^{(j)}\cap \hat{\mathcal{K}}^{(j)}$, $1/[((\hat{\mathbf{H}}^{(j)})^{\mathrm{H}}\hat{\mathbf{H}}^{(j)})^{-1}]_{kk}$ is a random variable with the PDF given as follows:
\begin{align}
g(x)= \frac{x^{\theta_{2}^{(j)}-1}e^{-x/2(\beta^2/(\beta+(\tau_{\infty}^{(j)})^2))}}{2^{\theta_{2}^{(j)}}\Gamma(\theta_{2}^{(j)})}, x\geq 0.
\end{align} 

\noindent Since $\gamma_{k|e^{(j)}f^{(j)}}^{(j)}$ is a scaled version of $\frac{1}{\left[\left((\hat{\mathbf{H}}^{(j)})^{\mathrm{H}}\hat{\mathbf{H}}^{(j)}\right)^{-1}\right]_{kk}}$, it is also Chi-square distributed, and its PDF can thus be obtained accordingly.
\end{proof}

With the conditional post-processing SNR distribution derived in Lemma \ref{lemma6}, the average BLER of an active user that is successfully detected in the $j$-th transmission block, is obtained in the following theorem.

\begin{theorem} \label{theorem3}
    The average BLER of an active user $k$ that is correctly detected in the $j$-th transmission block, i.e., $k\in\mathcal{K}^{(j)}\cap \hat{\mathcal{K}}^{(j)}$, is given as follows:
\begin{align}
\bar{\varepsilon}_{k}^{(j)}\!=\!\sum\limits_{f^{(j)}\!=0}^{N\!-\!K^{(j)}}\!{N\!-\!K^{(j)}\choose f^{(j)}} (P_{F}^{(j)})^{f^{(j)}}\!(1\!-\!P_{F}^{(j)})^{N\!-\!K^{(j)}\!-\!f^{(j)}}\!\sum\limits_{e^{(j)}=0}^{K^{(j)}} {K^{(j)}\choose e^{(j)}} (P_{M}^{(j)})^{e^{(j)}}(1\!-\!P_{M}^{(j)})^{K^{(j)}\!-\!e^{(j)}}\bar{\varepsilon}_{k|e^{(j)}f^{(j)}}^{(j)},
\end{align}

\noindent where 
\begin{align}
    \bar{\varepsilon}_{k|e^{(j)}f^{(j)}}^{(j)}\approx\left\{
    \begin{aligned}
    &\mathbb{E}_{\gamma_{k|e^{(j)}f^{(j)}}}\left[\varepsilon_{k}^{(j)}\left(\gamma_{k|e^{(j)}f^{(j)}}^{(j)}\right)\right], M\geq|\hat{\mathcal{K}}^{(j)}|,\\
    & 1,M<|\hat{\mathcal{K}}^{(j)}|.
    \end{aligned}
    \right.
\label{experror}
\end{align}
\end{theorem}

\begin{proof} For the case with $M\geq|\hat{\mathcal{K}}^{(j)}|$, conditioned on an outcome of user activity detection with $e^{(j)}$ missed detection and $f^{(j)}$ false alarm users, the average BLER of an active user that is detected correctly can be derived by taking expectation on both sides of (\ref{fblerror}) with respect to $\gamma_{k|e^{(j)}f^{(j)}}^{(j)}$. On the other hand, for $M<|\hat{\mathcal{K}}^{(j)}|$, $\bar{\varepsilon}_{k|e^{(j)}f^{(j)}}^{(j)}$ is approximated as 1 due to the assumption on the ZF equalizer that no data symbol can be correctly detected in this case. 

Then, the probability that $e^{(j)}$ users are miss-detected can be calculated as ${K^{(j)}\choose e^{(j)}} (P_{M}^{(j)})^{e^{(j)}}$ $(1-P_{M}^{(j)})^{K^{(j)}-e^{(j)}}$ since different active users can be deemed as statistically the same with the stochastic channel inversion power control. Similarly, the probability that there are $f^{(j)}$ false alarm users is given as ${N-K^{(j)}\choose f^{(j)}} (P_{F}^{(j)})^{f^{(j)}} (1-P_{F}^{(j)})^{N-K^{j}-f^{j}}$. Besides, the events that there are $e^{(j)}$ missed detection users and $f^{(j)}$ false alarm users are independent. By further incorporating the probabilities of different user activity detection outcomes, the average BLER of an active user that is correctly detected can be obtained.
%\begin{align}
%\begin{aligned}
%&+{N-K\choose f} P_{F}^{f} (1-P_{F})^{N-K}{K\choose e} P_{M}^{e}(1-%P_{M})^{K}\bar{\varepsilon}_{k|ef}+\cdots\\
%&=\sum\limits_{f=0}^{N-K}{N-K\choose f} P_{F}^{f} (1-P_{F})^{N-K-%f}\sum\limits_{e=0}^{K} {K\choose e} P_{M}^{e}(1-P_{M})^{K-%e}\bar{\varepsilon}_{k|ef}.
%\end{aligned}
%\end{align}
\end{proof}

Due to the statistical channel inversion power control, we have $\bar{\varepsilon}^{(j)} = \bar{\varepsilon}_{k}^{(j)}, k\in \mathcal{K}^{(j)}$. By substituting $\bar{\varepsilon}^{(j)}$, $P_{M}^{(j)}$, and $P_{F}^{(j)}$ into (\ref{bler}), the average BLER of an active user in the $j$-th transmission block is obtained, i.e., the BLER performance of different active users is identical, which can be denoted as $P^{(j)}_{e}$ by dropping the user index. Recall that $P^{\left(j\right)}_{e}\left( j \geq 2 \right)$ depends on both $K^{\left(j-1\right)}$ and $K^{\left(j\right)}$, we also denote it as $P^{\left(j\right)}_{e|K^{\left(j-1\right)}K^{\left(j\right)}}$ for convenience. For the initial transmission, we denote $P^{\left(1\right)}_{e}$ as $P^{\left(1\right)}_{e|K^{\left(1\right)}}$.

\subsection{Frame Error Rate of Grant-free Massive RA Systems with Retransmission}
Based on the BLER analysis in Section \ref{sectionvb}, the number of active users in the next transmission block follows a binomial distribution that can be easily characterized. Hence, the FER of grant-free massive RA systems with retransmission can be derived in a recursive form, as shown in the following theorem.
\begin{theorem} \label{theorem4}
    Define
\begin{align}
f_{\mathbb{K}^{(j+1)}}^{(J)}\triangleq \left\{
\begin{aligned}
&{{K^{(1)}-1}\choose {K^{(2)}-1}}(P_{e|K^{(1)}}^{(1)})^{K^{(2)}}(1-P_{e|K^{(1)}}^{(1)})^{K^{(1)}-K^{(2)}}, j=1,\\
&{{K^{(j)}-1}\choose {K^{(j+1)}-1}}(P_{e|K^{(j-1)}K^{(j)}}^{(j)})^{K^{(j+1)}}(1-P_{e|K^{(j-1)}K^{(j)}}^{(j)})^{K^{(j)}-K^{(j+1)}}, j= 2,\cdots, J-1,
\end{aligned}
\right.
\end{align}

\noindent where $\mathbb{K}^{(j)}\triangleq [K^{(1)},\cdots,K^{(j)}]$ denotes the numbers of active users in the first $j$ transmission blocks. For a given $\mathbb{K}^{(j)}$, the probability that the data of an active user in $\mathcal{K}^{\left(j\right)}$ cannot be decoded successfully from the $j$-th to the last transmission block, denoted as $\Phi_{\mathbb{K}^{(j)}}^{(J)}$, can be expressed in a recursive form as follows:
\begin{align}
\label{recursive}
\Phi_{\mathbb{K}^{(j)}}^{(J)}\triangleq \left\{
\begin{aligned}
&\sum\limits_{K^{(j+1)}=1}^{K^{(j)}} f_{\mathbb{K}^{(j+1)}}^{(J)}\Phi_{\mathbb{K}^{(j+1)}}^{(J)}, j=1,\cdots,J-1,\\
&P_{e|K^{(J-1)}K^{(J)}}^{(J)}, j= J, 
\end{aligned}
\right.
\end{align}

\noindent i.e., $\Phi_{\mathbb{K}^{(1)}}^{(J)}$ gives the FER of the grant-free massive RA system with $J$ transmission blocks.
\end{theorem}

\begin{proof}
Since all users are statistically the same with stochastic channel inversion power control, the FER of the grant-free massive RA system equals to the probability that the data of an arbitrary active user in $\mathcal{K}$ cannot be decoded successfully in all the $J$ transmission blocks. The proof of (\ref{recursive}) can be obtained via mathematical induction. First, $\Phi_{\mathbb{K}^{(J)}}^{(J)}=P_{e|K^{(J-1)}K^{(J)}}^{(J)}$ holds by definition. Then, suppose \vspace{-0.5cm}
\begin{align}
\Phi_{\mathbb{K}^{(j)}}^{(J)}= \sum\limits_{K^{(j+1)}=1}^{K^{(j)}} f_{\mathbb{K}^{(j+1)}}^{(J)}\Phi_{\mathbb{K}^{(j+1)}}^{(J)}.
\end{align}

\noindent To obtain $\Phi_{\mathbb{K}^{(j-1)}}^{(J)}$, it is equivalent to derive the probability that the data of an arbitrary active user in $\mathcal{K}^{\left(j-1\right)}$ (e.g. User A) cannot be decoded successfully from the $(j-1)$-th to the last transmission block given $\mathbb{K}^{(j)}$, which is the weighted average of the conditional probabilities for a given value of $K^{(j)}$. In other words, \vspace{-0.5cm}
\begin{align}
\Phi_{\mathbb{K}^{(j-1)}}^{(J)}= \sum\limits_{K^{(j)}=1}^{K^{(j-1)}} P_{K^{(j)}}\Phi_{\mathbb{K}^{(j)}}^{(J)},
\end{align}

\noindent where $P_{K^{(j)}}$ denotes the probability with $K^{(j)}$ active users (including User A) in the $j$-th transmission block. This is the probability that the data of User A is not decoded successfully in the $(j-1)$-th transmission block, meanwhile $K^{(j)}-1$ out of the remaining $K^{(j-1)}-1$ active users in the $(j-1)$-th transmission block are also in transmission error, i.e.,
\begin{align}
P_{K^{(j)}}=P_{e|K^{(j-2)}K^{(j-1)}}^{(j-1)} {{K^{(j-1)}-1}\choose {K^{(j)}-1}}(P_{e|K^{(j-2)}K^{(j-1)}}^{(j-1)})^{K^{(j)}-1}(1-P_{e|K^{(j-2)}K^{(j-1)}}^{(j-1)})^{K^{(j-1)}-K^{(j)}}= f_{\mathbb{K}^{(j)}}^{(J)}.
\end{align}
\noindent By substituting the expression of $P_{K^{\left(j\right)}}$ into $\Phi_{\mathbb{K}^{(j-1)}}^{(J)}$, we have 
\begin{align}
\Phi_{\mathbb{K}^{(j-1)}}^{(J)}= \sum\limits_{K^{(j)}=1}^{K^{(j-1)}}f_{\mathbb{K}^{(j)}}^{(J)} \Phi_{\mathbb{K}^{(j)}}^{(J)},
\end{align}
\noindent which ends the proof.
\end{proof}

Despite the FER has no closed-form expression, the recursive form in (\ref{recursive}) admits efficient evaluation. We also find the average BLER in (\ref{experror}) can be tightly approximated for $M\geq |\hat{\mathcal{K}}^{(j)}|$ as stated in the following corollary, which avoids computing the Q-function in the expectation operator so that the numerical evaluation can be further accelerated.

\begin{corollary}\label{corollary1}
Consider the case with $M\geq|\hat{\mathcal{K}}^{(j)}|$, and suppose there are  $e^{(j)}$ missed detection and $f^{(j)}$ false alarm users in the $j$-th transmission block, the average BLER in (\ref{experror}) can be approximated as follows: \vspace{-0.5cm}
\begin{align}
\bar{\varepsilon}_{k|e^{(j)}f^{(j)}}^{(j)}=\frac{\underline{\gamma}(\theta_{2}^{(j)},\frac{2^{\frac{c}{d}}-1}{2\theta_{1}^{(j)}})}{\Gamma(\theta_{2}^{(j)})},
\label{appro}
\end{align}

\noindent where $\underline{\gamma}(\cdot)$ denotes the lower incomplete Gamma function.
\end{corollary}

\begin{proof}
Please refer to Appendix \ref{appendixd}.
\end{proof}

We note that the FER achieved by using the conventional AMP-based receiver for all the transmission blocks is a special instance of (\ref{recursive}), where $\{P_{e|K^{\left(j-1\right)}K^{\left(j\right)}}^{(j)}\}$'s are replaced by the average BLER of the retransmission blocks with the AMP-based receiver that can be obtained via a similar method for Theorem \ref{theorem3}.
\vspace{-0.3cm}
\section{Simulation Results} \label{sectionvi}
We consider a single-cell uplink cellular network with 2,000 users, which are uniformly distributed within a circular ring with the inner and outer radius as 0.05 and 1 km, respectively. The BS antenna number $M$ and the pilot length $L$ are both 100, and the noise variance $\sigma^2$ is set as $-109$ dBm. We assume the blocklength $T$ equals 250 symbol intervals, and $c=50$ information bits need to be transmitted for each active user in each frame. By default, we simulate frames with two transmission blocks, i.e., $J=2$, and all the empirical results are averaged over $10,000$ independent channel realizations.
%$\beta$ & -123.8 dB
\vspace{-0.5cm}
\subsection{Results}
We first evaluate the activity detection error of the conventional AMP-based receiver and the proposed correlated AMP-based receiver in Fig. \ref{audampsi}, including the missed detection and false alarm probability versus the number of active users, where the target received signal strength is set as $\beta=-109$ dBm. Specifically, the activity detection error probability in the first transmission block is shown in Fig. \ref{audampsi}(a), while those in the second transmission block is shown in Fig. \ref{audampsi}(b)-(d), assuming $K^{(2)}\slash K= 0.3$, $0.5$, and $0.7$, respectively. Note that in the first transmission block, the operations of both receivers are the same so they have identical performance. It is observed that all the analytical curves match well with the empirical results, which corroborates the theoretical analysis of the missed detection and false alarm probability in (\ref{md}) and (\ref{fa}), respectively. For either the initial transmission or retransmission, both the missed detection and false alarm probability increase with the number of active users $K$, and among Fig. \ref{audampsi}(b)-(d), the error probabilities increase with $K^{(2)}$ for a given value of $K$. These observations indicate that a large number of concurrent transmissions degrade the activity detection performance due to more severe RA collision as non-orthogonal pilot sequences are adopted. Besides, the activity detection error probabilities of the correlated AMP-based receiver in the retransmission block are substantially smaller than those of the AMP-based receiver, which demonstrates the effectiveness of exploiting the user activity correlation among adjacent transmission blocks when designing receivers for grant-free massive RA systems with retransmission.
\begin{figure}[htbp]
\centering
\includegraphics[width=3in]{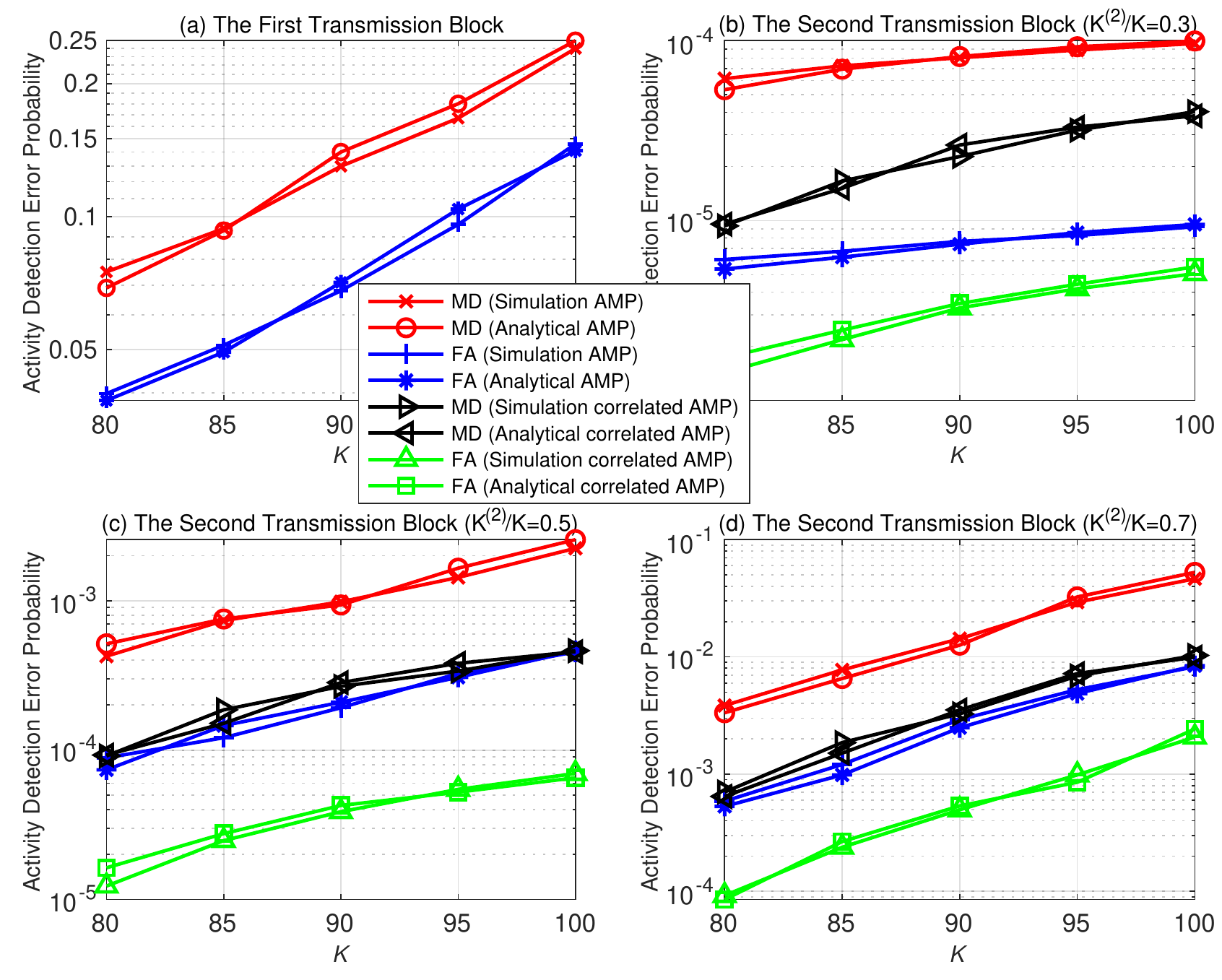}
\caption{Activity detection error probability vs. the number of active users ($J=2$ and $\beta = -109$ dBm), where ‘MD’ and ‘FA’ abbreviate missed detection and false alarm, respectively.}
\vspace{-0.5cm}
\label{audampsi}
\end{figure}

Next, we verify the channel estimation error as predicted by Lemma \ref{lemma4} for an active user that is correctly detected in Fig. \ref{delvk}. Similar to the observations in Fig. \ref{audampsi}, the simulation results are perfectly aligned with the theoretical analysis although the asymptotic assumptions on $N$, $K^{\left(j\right)}$, $M$, and $L$ may not be valid. Also, with an increased number of active users, the channel estimation performance of both the AMP and correlated AMP-based receivers deteriorates. Such a phenomenon accords to that appears in Fig. \ref{audampsi}, since both receivers determine the set of active users by using a thresholding operation on the estimated effective channel coefficients. Besides, comparisons of the channel estimation error in Fig. \ref{delvk}(b)-(d) further validate the competence of the correlated AMP-based receiver over the conventional design.
\begin{figure}[htbp]
\centering
\includegraphics[width=3in]{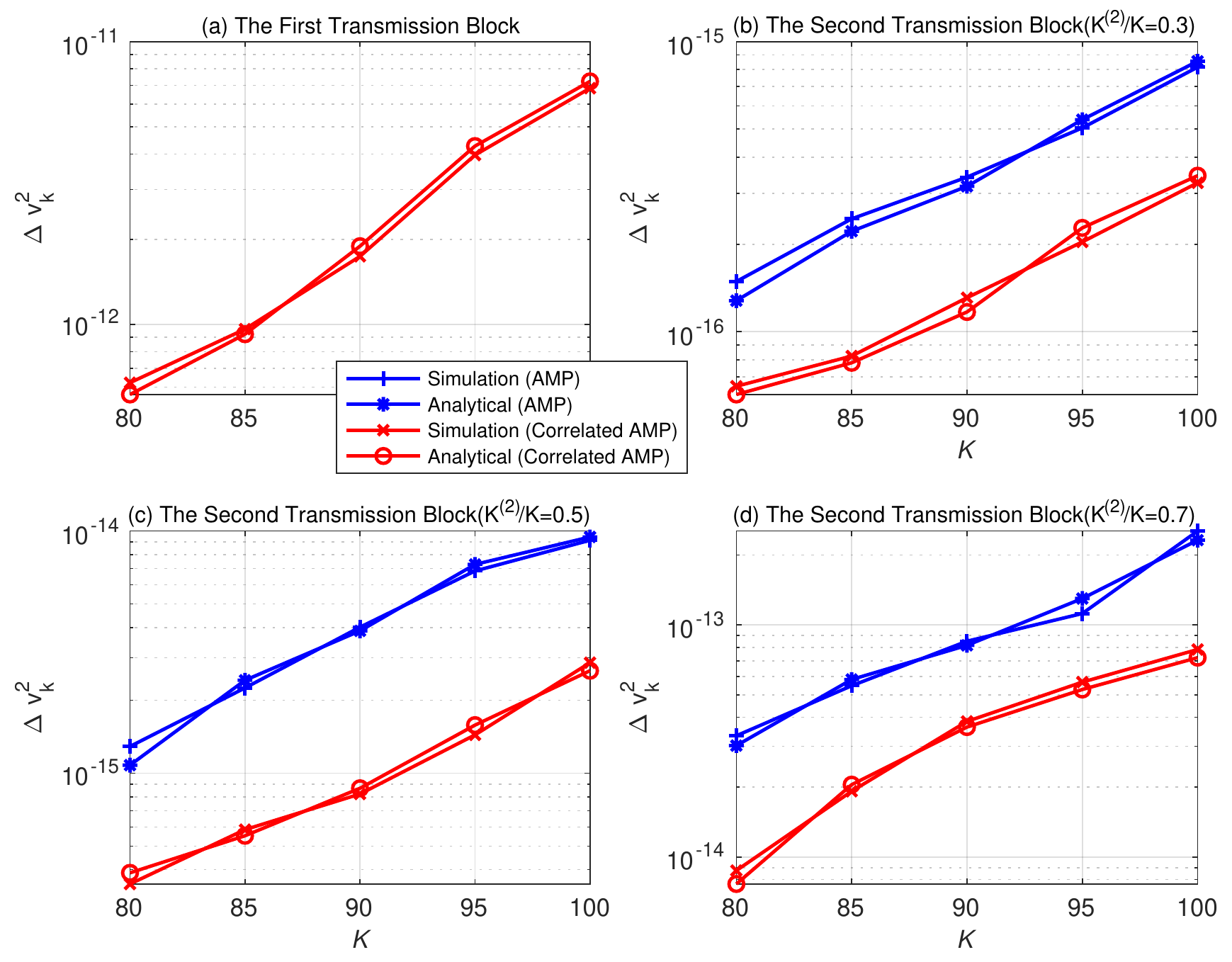}
\caption{Channel estimation error vs. the number of active users ($J=2$ and $\beta = -109$ dBm).}
\vspace{-0.9cm}
\label{delvk}
\end{figure}

The FERs achieved by the two receivers are evaluated in Fig. \ref{blervk}, where the theoretical results are computed using (\ref{recursive}), for which, two methods are adopted to evaluate the average BLER terms $\{\bar{\varepsilon}_{k|e^{(j)}f^{(j)}}^{(j)}\}$'s, including: 1) computing the expectation in (\ref{experror}) using numerical integration, and 2) computing the closed-form approximation in (\ref{appro}). We simulate two scenarios with different target received signal strength $\beta=-109$ dBm and $-106$ dBm. It is observed that the analytical curves obtained from both methods fit well with the empirical results, which shows the validity of the proposed closed-form average BLER approximation for efficient performance evaluation. Besides, as a direct consequence of the activity detection and channel estimation performance improvements, the correlated AMP-based receiver enjoys considerably lower FER compared with the AMP-based receiver. In addition, the FERs of both the AMP- and correlated AMP-based receiver decrease when the target received signal strength increases, and the correlated AMP-based receiver secures more significant performance gain compared with the AMP-based receiver with a higher target received signal strength.
\begin{figure}[htbp]
\centering
\includegraphics[width=3in]{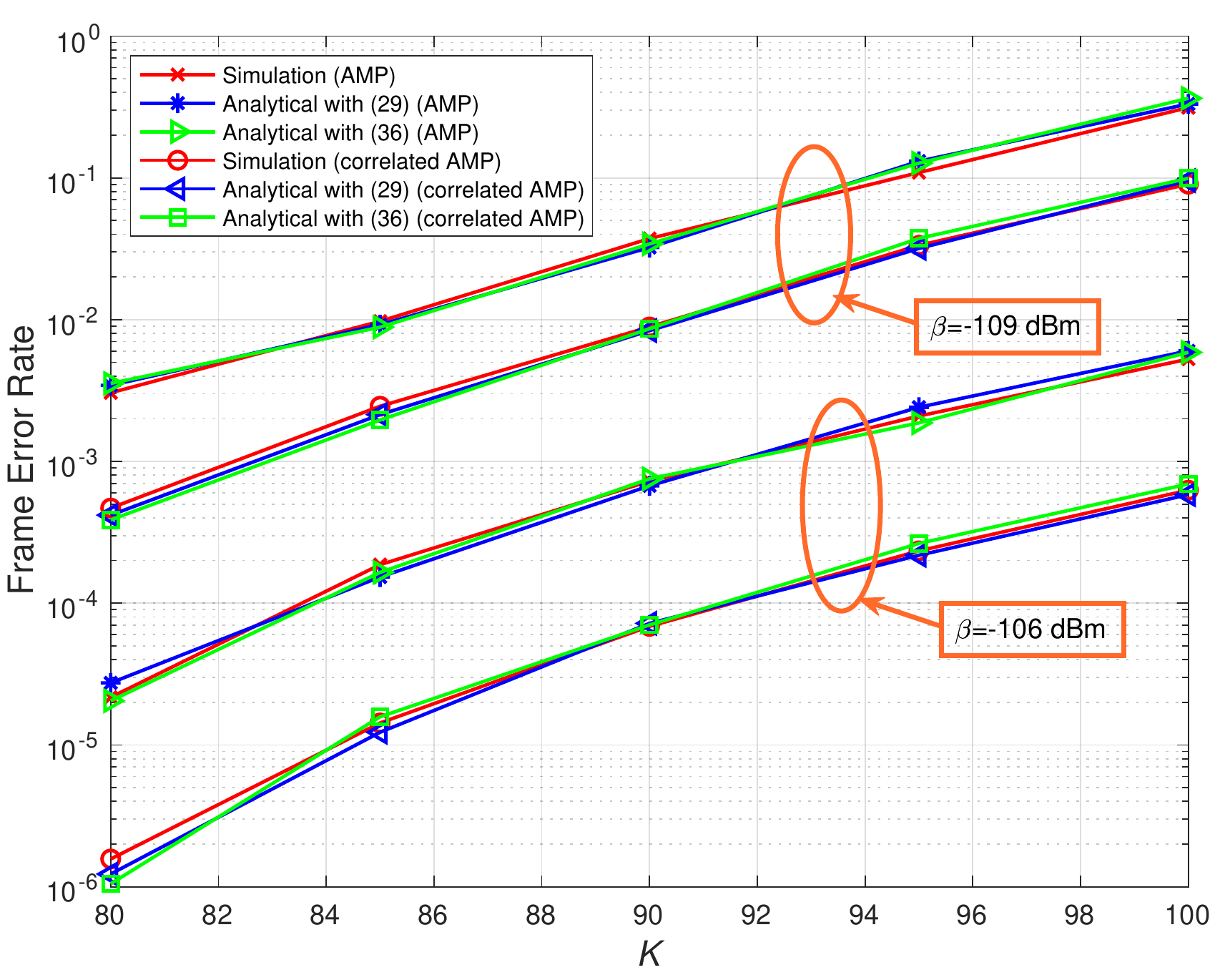}
\caption{Frame error rate vs. the number of active users ($J=2$, $\beta = -109$ dBm and $-106$ dBm).}
\vspace{-0.5cm}
\label{blervk}
\end{figure}

Fig. \ref{blervsnr} shows the FER of both the AMP- and correlated AMP-based receiver, versus the target receive SNR, where the analytical results are obtained with the average BLER approximation in (\ref{appro}). To demonstrate the benefits of retransmission, we consider three settings with different transmission blocks in a frame. We kindly note that obtaining stable simulation results at the high SNR regime requires extremely heavy computation, and the analytical results suffice to provide accurate performance characterization. It is seen from the figure that the FERs decrease with the target receive SNR, which comes at the expense of higher transmit power consumption at the users. Besides, the correlated AMP-based receiver consistently outperforms the AMP-based receiver and the performance gain is more significant with a higher receive SNR. This is attributed to the availability of more precise historical channel estimation results for the denoiser of the correlated AMP algorithm in the retransmission blocks. Besides, with more allowable retransmission attempts, the FERs decrease for a target receive SNR, which follows similar observations in grant-based RA because of the retransmission diversity gain. In addition, the FER reduction achieved by the correlated AMP-based receiver increases with $J$, which substantiates the need of dedicated receiver optimization for grant-free massive RA systems with retransmission, especially for those that can tolerate a longer delay, i.e., configuring a larger value of $J$ is feasible.
\vspace{-0.5cm}
\begin{figure}[htbp]
\centering
\includegraphics[width=3in]{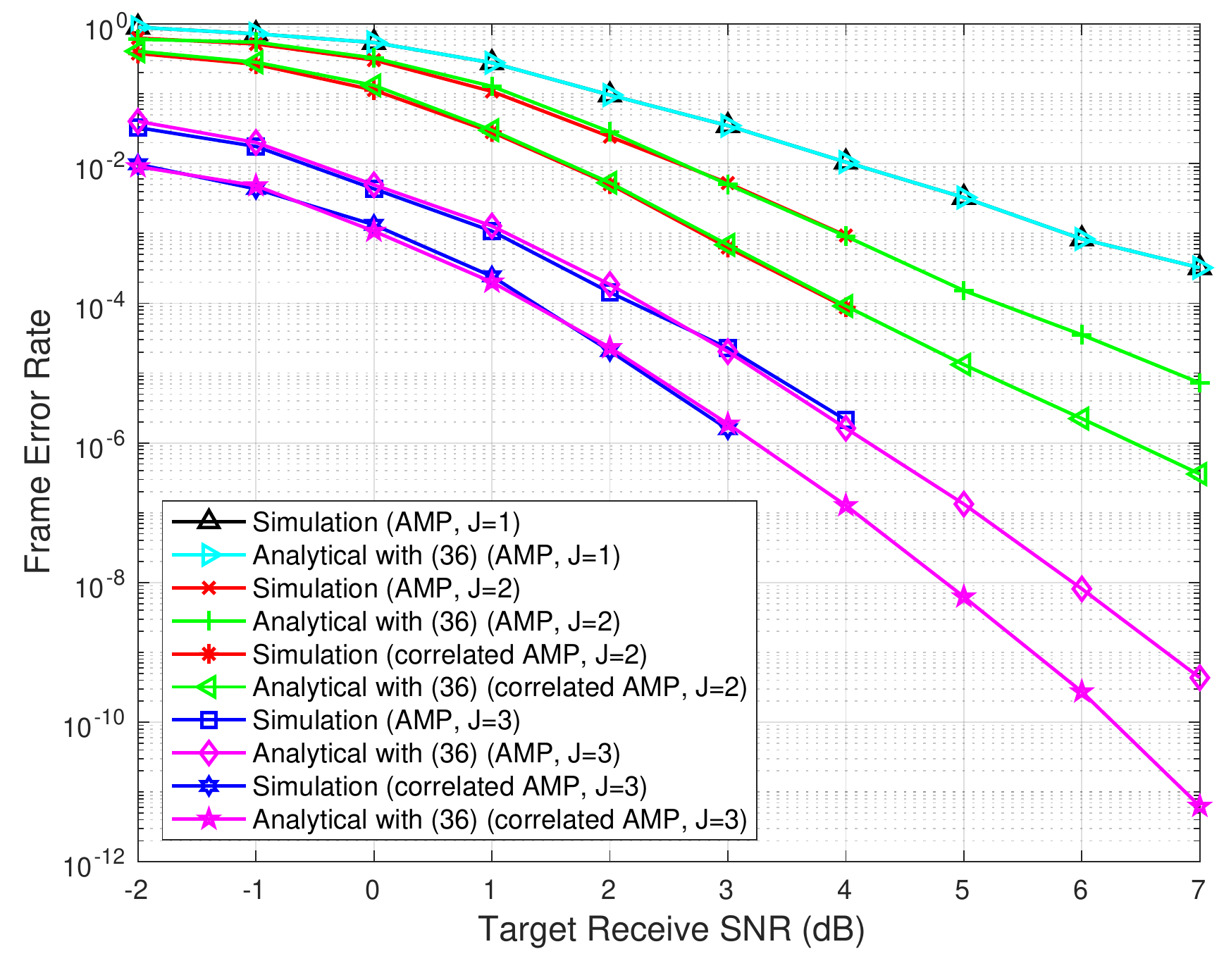}
\caption{Frame error rate vs. the target receive SNR.}
\label{blervsnr}
\end{figure}
\section{Conclusions \label{sectionvii}}
In this paper, we investigated the grant-free massive random access (RA) with retransmission to improve the reliability of massive machine-type communications. Observing the user activity correlation between adjacent transmission blocks, we developed a new uplink receiver at the base station based on the correlated approximate message passing (AMP) algorithm, which effectively utilizes such prior information to improve the accuracy of user activity detection, channel estimation, and data decoding. Then, based on the state evolution of the correlated AMP algorithm, we further analyzed the frame error rate by characterizing the activity detection, channel estimation, and data decoding error in a single transmission block. From our simulation, the proposed correlated AMP-based receiver showed its effectiveness in grant-free massive RA systems with retransmission compared with the classic AMP-based receiver, and the corresponding theoretical analysis is also validated.

From our study, it is concluded that retransmission is an important mechanism for grant-free massive RA, where the receiver operations at the retransmission blocks should be carefully designed by exploiting all the available prior knowledge. Besides, the proposed analytical framework is valuable for investigating other design problems for grant-free massive RA systems with retransmission, for instance, opportunistic user admission control, to further improve the communication performance.

\appendices
\vspace{-0.5cm}
\section{Proof of Theorem 1 \label{appendixa}}
We derive the MMSE denoiser based on the correlation of two adjacent transmission blocks, which is given by the conditional expectation of $\mathbf{X}^{(j)}$ as follows:
\begin{align}
\eta^{(j)}_{t}(\mathbf{a}_{n,t}^{(j)},\mathbf{b}_{n,\infty}^{(j-1)})=\mathbb{E}[\mathbf{X}^{(j)}|\mathbf{X}^{(j)}+(\mathbf{\Sigma}^{(j)}_{t})^{\frac{1}{2}}\mathbf{V}^{(j)}=\mathbf{a}_{n,t}^{(j)}, \tilde{\mathbf{X}}^{(j-1)}=\mathbf{b}_{n,\infty}^{(j-1)}].
\end{align}

\noindent The above expression can be expanded using the law of total probability as follows:
\begin{align}
\begin{aligned}
\eta^{(j)}_{t}(\mathbf{a}_{n,t}^{(j)},\mathbf{b}_{n,\infty}^{(j-1)})\ &=\sum_{i=1}^{3}p(\text{Case} \ i|\mathbf{a}_{n,t}^{(j)},\mathbf{b}_{n,\infty}^{(j-1)})\mathbb{E}[\mathbf{h}^{(j)}_{n}|\mathbf{a}_{n,t}^{(j)},\mathbf{b}_{n,\infty}^{(j-1)},\text{Case} \ i]\\
&\overset{(a)}{=}p(\text{Case} \ 3|\mathbf{a}_{n,t}^{(j)},\mathbf{b}_{n,\infty}^{(j-1)})\mathbb{E}[\mathbf{h}^{(j)}_{n}|\mathbf{a}_{n,t}^{(j)},\mathbf{b}_{n,\infty}^{(j-1)},\text{Case} \ 3],
\end{aligned}
\label{etacal}
\end{align}
\noindent where (a) is because $\mathbb{E}[\mathbf{h}^{(j)}_{n}|\mathbf{a}_{n,t}^{(j)},\mathbf{b}_{n,\infty}^{(j-1)},\text{Case 1}]=\mathbb{E}[\mathbf{h}^{(j)}_{n}|\mathbf{a}_{n,t}^{(j)},\mathbf{b}_{n,\infty}^{(j-1)},\text{Case 2}]=0$ since $u_n^{(j)}=0$ in both Case 1 and Case 2. It remains to derive $p(\text{Case} \ 3|\mathbf{a}_{n,t}^{(j)},\mathbf{b}_{n,\infty}^{(j-1)})$ and $\mathbb{E}[\mathbf{h}^{(j)}_{n}|\mathbf{a}_{n,t}^{(j)},\mathbf{b}_{n,\infty}^{(j-1)},\text{Case} \ 3]$. By applying the Bayes' theorem, $p(\text{Case} \ 3|\mathbf{a}_{n,t}^{(j)},\!\mathbf{b}_{n,\infty}^{(j-1)})$ can be written as follows:
\begin{align}
\begin{aligned}
p(\text{Case} \ 3|\mathbf{a}_{n,t}^{(j)},\mathbf{b}_{n,\infty}^{(j-1)})&=\frac{p(\mathbf{a}_{n,t}^{(j)},\mathbf{b}_{n,\infty}^{(j-1)}|\text{Case} \ 3)p(\text{Case} \ 3)}{p(\mathbf{a}_{n,t}^{(j)},\mathbf{b}_{n,\infty}^{(j-1)})}\\
&=\frac{p(\mathbf{a}_{n,t}^{(j)},\mathbf{b}_{n,\infty}^{(j-1)}|\text{Case} \ 3)p(\text{Case} \ 3)}{\sum_{i=1}^{3} p(\mathbf{a}_{n,t}^{(j)},\mathbf{b}_{n,\infty}^{(j-1)}|\text{Case} \ i)p(\text{Case} \ i)},
\end{aligned}
\label{pcase3}
\end{align}

\noindent where $p(\text{Case} \ i)=\epsilon_{i}^{(j)}$, $\forall i$. In particular,
\begin{align}
\begin{aligned}
p(\mathbf{a}_{n,t}^{(j)},\mathbf{b}_{n,\infty}^{(j-1)}|\text{Case} \ 1)&=p(\mathbf{a}_{n,t}^{(j)}|\text{Case} \ 1)p(\mathbf{b}_{n,\infty}^{(j-1)}|\text{Case} \ 1)\\
&=\frac{1}{\pi(\tau^{(j)}_{t})^{2M} } e^{ -\frac{||\mathbf{a}_{n,t}^{(j)}||^2 }{(\tau^{(j)}_{t})^2}} \frac{1}{\pi (\tau^{(j-1)}_{\infty})^{2M}} e^{ -\frac{||\mathbf{b}_{n,\infty}^{(j-1)}||^2 }{(\tau^{(j-1)}_{\infty})^{2}}},
\end{aligned}
\end{align}

\noindent which is due to the fact that $\mathbf{a}_{n,t}^{(j)}$ and $\mathbf{b}_{n,\infty}^{(j-1)}$ are independent, and $\mathbf{h}_{n}^{(j)}=\mathbf{h}_{n}^{(j-1)}=0$. Similarly,
%\begin{align}
%p(\mathbf{a}_{n,t}^{r},\mathbf{b}_{n,\infty}^{o}|\text{Case} \ 2)=\frac{1}{\pi (\tau^{r}_{t})^{2M}} e^{ -\frac{||\mathbf{a}_{n,t}^{r}||^2 }{(\tau^{r}_{t})^2}} \frac{1}{\pi \left(\beta+(\tau^{o}_{\infty})^2\right)^{M}} e^{ -\frac{||\mathbf{b}_{n,\infty}^{o}||^2 }{\beta+(\tau^{o}_{\infty})^2}},
%\end{align}
\begin{align}
p(\mathbf{a}_{n,t}^{(j)},\mathbf{b}_{n,\infty}^{(j-1)}|\text{Case} \ 2)=\frac{1}{\pi \left(\beta+(\tau^{(j)}_{t})^2\right)^{M}} e^{ -\frac{||\mathbf{a}_{n,t}^{(j)}||^2 }{\beta+(\tau^{(j)}_{t})^{2}}} \frac{1}{\pi (\tau^{(j-1)}_{\infty})^{2M}} e^{ -\frac{||\mathbf{b}_{n,\infty}^{(j-1)}||^{2} }{(\tau^{(j-1)}_{\infty})^{2}}},
\end{align}
\begin{align}
p(\mathbf{a}_{n,t}^{(j)},\mathbf{b}_{n,\infty}^{(j-1)}|\text{Case} \ 3)=\frac{1}{\pi \left(\beta+(\tau^{(j)}_{t})^{2}\right)^{M}} e^{ -\frac{||\mathbf{a}_{n,t}^{(j)}||^2 }{\beta+(\tau^{(j)}_{t})^{2}}} \frac{1}{\pi \left(\beta+(\tau^{(j-1)}_{\infty})^{2}\right)^{M}} e^{ -\frac{||\mathbf{b}_{n,\infty}^{(j-1)}||^2 }{\beta+(\tau^{(j-1)}_{\infty})^2}}.
\end{align}

Since $\mathbf{b}_{n,\infty}^{(j-1)}$ is independent of $\mathbf{h}_{n}^{(j)}$, we have
%\begin{align}
%\mathbb{E}[\mathbf{h}^{r}_{n}|\mathbf{a}_{n,t}^{r},\mathbf{b}_{n,\infty}^{o},\text{Case 2}]=\mathbb{E}[\mathbf{h}^{r}_{n}|\mathbf{a}_{n,t}^{r},\text{Case 2}]=\frac{\beta\mathbf{a}_{n,t}^{r}}{\beta+(\tau_{t}^{r})^2}.
%\end{align}
\begin{align}
\mathbb{E}[\mathbf{h}^{(j)}_{n}|\mathbf{a}_{n,t}^{(j)},\mathbf{b}_{n,\infty}^{(j-1)},\text{Case 3}]=\mathbb{E}[\mathbf{h}^{(j)}_{n}|\mathbf{a}_{n,t}^{(j)},\text{Case 3}]=\frac{\beta\mathbf{a}_{n,t}^{(j)}}{\beta+(\tau_{t}^{(j)})^2}.
\label{expcase3}
\end{align}

By substituting (\ref{pcase3})-(\ref{expcase3}) into (\ref{etacal}), the MMSE denoiser as stated in (\ref{denoiserampsi}) can be obtained.
\vspace{-0.5cm}
\section{Proof of Theorem 2 \label{appendixb}}
When the correlated AMP algorithm converges, the hypothesis testing problem is formulated to determine the set of active users as follows:
\begin{align}
\left\{\begin{array}{l}
H_{0}: \mathbf{h}^{(j)}_{n}=0, \text {inactive,} \\
H_{1}: \mathbf{h}^{(j)}_{n} \neq 0, \text { active. }
\end{array}\right.
\end{align}

\noindent Since the observations $\mathbf{h}^{(j)}_{n,\infty}+\left(\mathbf{R}^{(j)}_{\infty}\right)^{\mathrm{H}}\mathbf{p}_{n}$ and $\mathbf{h}^{(j-1)}_{n,\infty}+\left(\mathbf{R}^{(j-1)}_{\infty}\right)^{\mathrm{H}}\mathbf{p}_{n}$ that can be statistically characterized by $\mathbf{a}_{n,\infty}^{(j)}$ and $\mathbf{b}_{n,\infty}^{(j-1)}$ from the state evolution are available, the decision rule can be expressed in terms of the likelihood functions of the two cases, i.e., $p\left(\mathbf{a}_{n,\infty}^{(j)},\mathbf{b}_{n,\infty}^{(j-1)} \mid \mathbf{h}^{(j)}_{n} \neq 0\right)$ and $p\left(\mathbf{a}_{n,\infty}^{(j)},\mathbf{b}_{n,\infty}^{(j-1)} \mid \mathbf{h}^{(j)}_{n} = 0\right)$, as follows: \vspace{-0.3cm}
\begin{align}
\ln \left(\frac{p\left(\mathbf{a}_{n,\infty}^{(j)},\mathbf{b}_{n,\infty}^{(j-1)} \mid \mathbf{h}^{(j)}_{n} \neq 0\right)}{p \left(\mathbf{a}_{n,\infty}^{(j)},\mathbf{b}_{n,\infty}^{(j-1)} \mid \mathbf{h}^{(j)}_{n}=0\right)}\right) \overset{H_{0}}{\underset{H_{1}}{\lessgtr}}0.
\label{decisionrule}
\end{align}

\noindent By applying the Bayes' theorem, $p\left(\mathbf{a}_{n,\infty}^{(j)},\mathbf{b}_{n,\infty}^{(j-1)} \mid \mathbf{h}^{(j)}_{n} \neq 0\right)$ and $p\left(\mathbf{a}_{n,\infty}^{(j)},\mathbf{b}_{n,\infty}^{(j-1)} \mid \mathbf{h}^{(j)}_{n} = 0\right)$ can be derived as follows:
\begin{align}
\begin{aligned}
p\left(\mathbf{a}_{n,\infty}^{(j)},\mathbf{b}_{n,\infty}^{(j-1)} \mid \mathbf{h}^{(j)}_{n} \neq 0\right)&=\frac{p\left(\mathbf{a}_{n,\infty}^{(j)},\mathbf{b}_{n,\infty}^{(j-1)},\mathbf{h}^{(j)}_{n} \neq 0\right)}{p(\mathbf{h}^{(j)}_{n} \neq 0)}=\frac{p(\mathbf{a}_{n,\infty}^{(j)},\mathbf{b}_{n,\infty}^{(j-1)}|\text{Case} \ 3)\epsilon_3^{(j)}}{\epsilon_3^{(j)}}\\
&=\frac{1}{\pi \left(\beta+(\tau^{(j)}_{\infty})^2\right)^{M}} e^{ -\frac{||\mathbf{a}_{n,\infty}^{(j)}||^2 }{\beta+(\tau^{(j)}_{\infty})^2}} \frac{1}{\pi \left(\beta+(\tau^{(j-1)}_{\infty})^2\right)^{M}} e^{ -\frac{||\mathbf{b}_{n,\infty}^{(j-1)}||^2 }{\beta+(\tau^{(j-1)}_{\infty})^2}},
\end{aligned}
\end{align}

\begin{align}
\begin{aligned}
p\left(\mathbf{a}_{n,\infty}^{(j)},\mathbf{b}_{n,\infty}^{(j-1)} \mid \mathbf{h}^{(j)}_{n} =0\right)
&=\frac{p(\mathbf{a}_{n,\infty}^{(j)},\mathbf{b}_{n,\infty}^{(j-1)}|\text{Case} \ 1)\epsilon_1^{(j)}+p(\mathbf{a}_{n,\infty}^{(j)},\mathbf{b}_{n,\infty}^{(j-1)}|\text{Case} \ 2)\epsilon_2^{(j)}}{\epsilon_1^{(j)}+\epsilon_2^{(j)}}\\
&=\frac{\frac{1}{ \pi (\tau^{(j-1)}_{\infty})^{2M}} e^{-\frac{||\mathbf{b}_{n,\infty}^{(j-1)}||^2 }{(\tau^{(j-1)}_{\infty})^2}}\left(\frac{\epsilon_1^{(j)}}{\pi (\tau^{(j)}_{\infty})^{2M}} e^{ -\frac{||\mathbf{a}_{n,\infty}^{(j)}||^2}{(\tau^{(j)}_{\infty})^2}}+\frac{\epsilon_{2}^{(j)}}{\pi \left(\beta+(\tau^{(j)}_{\infty})^2\right)^{M}} e^{ -\frac{||\mathbf{a}_{n,\infty}^{(j)}||^2 }{\beta+(\tau^{(j)}_{\infty})^2}}\right)}{\epsilon_1^{(j)}+\epsilon_2^{(j)}}.
\end{aligned}
\end{align}
\noindent Therefore, \vspace{-0.2cm}
\begin{align}
\frac{p\left(\mathbf{a}_{n,t}^{(j)},\mathbf{b}_{n,\infty}^{(j-1)} \mid \mathbf{h}^{(j)}_{n} \neq 0\right)}{p \left(\mathbf{a}_{n,t}^{(j)},\mathbf{b}_{n,\infty}^{(j-1)} \mid \mathbf{h}^{(j)}_{n}=0\right)}=\frac{\epsilon_{1}^{(j)}+\epsilon_{2}^{(j)}}{\Phi_{n,2}^{(j-1)}\left(\epsilon_{1}^{(j)}\Phi_{n,\infty,1}^{(j)}+\epsilon_{2}^{(j)}\right)}.
\label{beforelog}
\end{align}

\noindent By taking the logarithm on both sides of (\ref{beforelog}), the decision rule in (\ref{decisionrule}) can be equivalently expressed as follows: \vspace{-0.3cm}
\begin{align}
||\mathbf{a}_{n,\infty}^{(j)}||^2 \overset{H_{0}}{\underset{H_{1}}{\lessgtr}} \frac{M \ln \left(1+\frac{\beta}{(\tau^{(j)}_{\infty})^{2}}\right)+\ln\left(\frac{\epsilon_{1}^{(j)}\Phi_{n,2}^{(j-1)}}{\epsilon_{1}^{(j)}+\epsilon_{2}^{(j)}-\epsilon_{2}^{(j)}\Phi_{n,2}^{(j-1)}}\right)}{\frac{1}{(\tau_{\infty}^{(j)})^2}-\frac{1}{(\tau_{\infty}^{(j)})^{2}+\beta}}.
\end{align}
\noindent To further derive a common threshold, we use $\mathbb{E}[||\mathbf{\tilde{X}}^{(j-1)}||^{2}]$ to replace $||\mathbf{b}_{n,\infty}^{(j-1)}||^{2}$ in $\Phi_{n,2}^{(j-1)}$, and the threshold can be derived as follows:
\begin{align}
||\mathbf{a}_{n,\infty}^{(j)}||^2 \overset{H_{0}}{\underset{H_{1}}{\lessgtr}} \frac{M \ln \left(1+\frac{\beta}{(\tau^{(j)}_{\infty})^{2}}\right)+\ln\left(\frac{\epsilon_{1}^{(j)}\Phi_{2}^{(j-1)}}{\epsilon_{1}^{(j)}+\epsilon_{2}^{(j)}-\epsilon_{2}^{(j)}\Phi_{2}^{(j-1)}}\right)}{\frac{1}{(\tau_{\infty}^{(j)})^2}-\frac{1}{(\tau_{\infty}^{(j)})^{2}+\beta}}\triangleq l_{r}^{(j)},
\label{beforeactiveampsi}
\end{align}

\noindent where $\Phi_{2}^{(j-1)}=\left(\frac{\beta+(\tau_{\infty}^{(j-1)})^2}{(\tau_{\infty}^{(j-1)})^2}\right)^M e^ {\left(\frac{1}{\beta+(\tau_{\infty}^{(j-1)})^2}-\frac{1}{(\tau_{\infty}^{(j-1)})^2}\right)\left(\beta\lambda^{(j-1)}+(\tau_{\infty}^{(j-1)})^{2}\right)}$. Since $\mathbf{a}_{n,\infty}^{(j)}$ is statistically equivalent to $\mathbf{h}^{(j)}_{n,\infty}+\left(\mathbf{R}^{(j)}_{\infty}\right)^{\mathrm{H}}\mathbf{p}_{n}$, by replacing it with $\mathbf{h}^{(j)}_{n,\infty}+\left(\mathbf{R}^{(j)}_{\infty}\right)^{\mathrm{H}}\mathbf{p}_{n}$ in (\ref{beforeactiveampsi}), we complete the proof.
\vspace{-0.5cm}
\section{Proof of Lemma 3 \label{appendixc}}
For the retransmission blocks, i.e., $j>1$, according to the state evolution of the correlated AMP algorithm, $\mathbf{h}^{(j)}_{n,t}+\left(\mathbf{R}_{t}^{(j)}\right)^{\mathrm{H}} \mathbf{p}_{n}^{(j)}$ is statistically equivalent to $\mathbf{a}_{n,t}^{(j)}=\mathbf{h}^{(j)}_{n}+\tau^{(j)}_{t}\mathbf{v}^{(j)}_{n}$. Thus, when $u_{n}^{(j)}=0$, entries of $\mathbf{h}^{(j)}_{n,t}+\left(\mathbf{R}_{t}^{(j)}\right)^{\mathrm{H}} \mathbf{p}_{n}$ follow a independent and identical complex Gaussian distribution with zero mean and variance $\tau_{t}^{(j)}$. When $u_{n}^{(j)}=1$, entries of $\mathbf{h}^{(j)}_{n,t}+\left(\mathbf{R}_{t}^{(j)}\right)^{\mathrm{H}} \mathbf{p}_{n}$ follow a independent and identical complex Gaussian distribution with zero mean and variance $\beta+\tau_{t}^{(j)}$. As a result, random variables $||\mathbf{h}^{(j)}_{n,t}+\left(\mathbf{R}_{t}^{(j)}\right)^{\mathrm{H}} \mathbf{p}_{n}||^{2}/2(\tau_{t}^{(j)})^2$ (corresponds to $u_{n}^{(j)}=0$) and $||\mathbf{h}^{(j)}_{n,t}+\left(\mathbf{R}_{t}^{(j)}\right)^{\mathrm{H}} \mathbf{p}_{n}||^{2}/2(\beta+(\tau_{t}^{(j)})^2)$ (corresponds to $u_{n}^{(j)}=1$) follow the $\chi^2$ distribution with $2M$ degrees of freedom. Let $X$ be a random variable following the $\chi^2$ distribution with $2M$ degrees of freedom. According to the user activity detection criteria in (\ref{activeampsi}), when the correlated AMP algorithm converges, the probability of missed detection can be obtained as follows: \vspace{-0.1cm}
\begin{align}
\begin{aligned}
P_{M}^{(j)} &=p\left(X \leq \frac{2 M l_{r}^{(j)}}{\beta+(\tau_{\infty}^{(j)})^{2}}\right)\\
&=\frac{1}{\Gamma(M)} \bar{\gamma}\left(M, \frac{(\tau^{(j)}_{\infty})^{2}}{\beta}\left(M\ln\left(1+\frac{\beta}{(\tau^{(j)}_{\infty})^{2}}\right)+\ln\left(\frac{\epsilon_{1}^{(j)}\Phi_{2}^{(j-1)}}{\epsilon_{1}^{(j)}+\epsilon_{2}^{(j)}-\epsilon_{2}^{(j)}\Phi_{2}^{(j-1)}}\right)\right)\right).
\end{aligned}
\end{align}

\noindent Similarly, the probability of missed detection can be expressed as follows:
\begin{align}
    \begin{aligned}
P_{F}^{(j)}&=p\left(X \geq \frac{2 M l_{r}^{(j)}}{(\tau_{\infty}^{(j)})^{2}}\right)\\
&=1-\frac{1}{\Gamma(M)} \bar{\gamma}\left(M, \left(1+ \frac{(\tau^{(j)}_{\infty})^{2}}{\beta}\right)\left(M\ln\left(1+\frac{\beta}{(\tau^{(j)}_{\infty})^{2}}\right)+\ln\left(\frac{\epsilon_{1}^{(j)}\Phi_{2}^{(j-1)}}{\epsilon_{1}^{(j)}+\epsilon_{2}^{(j)}-\epsilon_{2}^{(j)}\Phi_{2}^{(j-1)}}\right)\right)\right).
\end{aligned}
\end{align}

For the first transmission block ($j=1$) that uses the AMP algorithm, the probabilities of missed detection and false alarm can be similarly derived \cite{lliu2018}.

\section{Proof of Corollary 1 \label{appendixd}}
By using the linearization technique, the error rate $\varepsilon_{k}^{(j)}(\gamma_{k}^{(j)})$ in (\ref{fblerror}) can be tightly approximated as follows \cite{b19}: \vspace{-0.3cm}
\begin{align}
\varepsilon_{k}^{(j)}(\gamma_{k}^{(j)})=Q\left(\frac{C(\gamma_{k}^{(j)})-R}{\sqrt{V(\gamma_{k}^{(j)})/d}}\right) \approx\left\{\begin{array}{cc}
1, & \gamma_{k}^{(j)} \leq v \\
A^{(j)}\left(\gamma_{k}^{(j)}\right), & v <\gamma_{k}^{(j)}<\mu \\
0, & \gamma_{k}^{(j)} \geq \mu
\end{array}\right.,
\label{approximation}
\end{align}

\noindent where $A^{(j)}\left(\gamma_{k}^{(j)}\right)\triangleq \frac{1}{2}-\chi\sqrt{d}(\gamma_{k}^{(j)}-r)$, $\chi\triangleq \sqrt{\frac{1}{2\pi (2^{\frac{2c}{d}}-1)}}$, $v\triangleq r-\frac{1}{2\chi\sqrt{d}}$, $\mu\triangleq r+\frac{1}{2\chi\sqrt{d}}$, and $r\triangleq 2^{R}-1=2^{\frac{c}{d}}-1$. By substituting the right-hand side of (\ref{approximation}) into (\ref{experror}) with $\gamma_{k}^{(j)} = \gamma_{k|e^{(j)}f^{(j)}}^{(j)}$, $\bar{\varepsilon}_{k|e^{(j)}f^{(j)}}^{(j)}$ can be approximated as follows: \vspace{-0.5cm}
\begin{align}
\bar{\varepsilon}_{k|e^{(j)}f^{(j)}}^{(j)} \approx \chi \sqrt{d} \int_{v}^{\mu} F_{\gamma_{k|e^{(j)}f^{(j)}}^{(j)}}(x) dx,
\end{align}
\noindent where $F_{\gamma_{k|e^{(j)}f^{(j)}}^{(j)}}(x)\triangleq \frac{\underline{\gamma}(\theta_{2}^{(j)},\frac{x}{2\theta_{1}^{(j)}})}{\Gamma(\theta_{2}^{(j)})}$ is the cumulative distribution function (CDF) of the post-processing SNR. Typically, since $d>6\pi$ can be satisfied easily, the integral interval $u-v=\frac{1}{\chi\sqrt{d}}=\sqrt{\frac{2\pi (2^{\frac{2c}{d}}-1)}{d}} < \sqrt{\frac{6\pi}{d}}$ is small, e.g., no larger than $0.5$ in our simulation setting. Therefore, we apply the ﬁrst-order Riemann integral approximation, i.e., $\int_{a}^{b} f(x) d x \approx(b-a) f\left(\frac{a+b}{2}\right)$, to calculate the integral value \cite{yyu2018}. Therefore, the first case in (\ref{experror}) can be approximated as follows: \vspace{-0.4cm}
\begin{align}
\bar{\varepsilon}_{k|e^{(j)}f^{(j)}}^{(j)} \approx \chi \sqrt{d} (\mu-v) F_{\gamma_{k|e^{(j)}f^{(j)}}^{(j)}}\left(\frac{\mu+v}{2}\right)= \frac{\underline{\gamma}(\theta_{2}^{(j)},\frac{2^{\frac{c}{d}}-1}{2\theta_{1}^{(j)}})}{\Gamma(\theta_{2}^{(j)})}.
\end{align}
\ifCLASSOPTIONcaptionsoff
  \newpage
\fi

\end{document}